\documentclass[final,twoside,11pt]{entics} 

\usepackage{enticsmacro}

\usepackage[utf8]{inputenc}
\usepackage[T5,T1]{fontenc}
\usepackage[english]{babel}

\usepackage{enumitem}
\usepackage{multicol}
\usepackage{amsmath}
\usepackage{amssymb}
\usepackage{cleveref}
\usepackage{bussproofs}
\usepackage{csquotes}
\usepackage{mdframed}
\usepackage{mathpartir}
\usepackage{url}
\usepackage{mdframed}
\usepackage{mathbbol}
\usepackage{mathtools}
\usepackage{thmtools}
\usepackage{mathpartir}
\usepackage{thm-restate}
\usepackage{bm,accents}
\usepackage{appendix}
\usepackage{xfrac}
\usepackage{relsize}
\usepackage{faktor}
\usepackage{cmll}
\usepackage{array}
\usepackage[all]{xy}
\usepackage{centernot}
\usepackage{stmaryrd}
\usepackage{cmll}
\usepackage{xypic}
\usepackage{formal-grammar}
\usepackage{subcaption}
\usepackage{quiver}
\DeclareSymbolFont{AMSb}{U}{msb}{m}{n}
\DeclareSymbolFontAlphabet{\mathbb}{AMSb}
\DeclareSymbolFontAlphabet{\mathbbl}{bbold}

\usepackage{tikz}
\usetikzlibrary{positioning}

\ifdefined\cA
\else
\newcommand\cA{\mathcal{A}}
\fi
\ifdefined\cB
\else
\newcommand\cB{\mathcal{B}}
\fi
\ifdefined\cC
\else
\newcommand\cC{\mathcal{C}}
\fi
\ifdefined\cD
\else
\newcommand\cD{\mathcal{D}}
\fi
\ifdefined\cE
\else
\newcommand\cE{\mathcal{E}}
\fi
\ifdefined\cF
\else
\newcommand\cF{\mathcal{F}}
\fi
\ifdefined\cG
\else
\newcommand\cG{\mathcal{G}}
\fi
\ifdefined\cH
\else
\newcommand\cH{\mathcal{H}}
\fi
\ifdefined\cI
\else
\newcommand\cI{\mathcal{I}}
\fi
\ifdefined\cJ
\else
\newcommand\cJ{\mathcal{J}}
\fi
\ifdefined\cK
\else
\newcommand\cK{\mathcal{K}}
\fi
\ifdefined\cL
\else
\newcommand\cL{\mathcal{L}}
\fi
\ifdefined\cM
\else
\newcommand\cM{\mathcal{M}}
\fi
\ifdefined\cN
\else
\newcommand\cN{\mathcal{N}}
\fi
\ifdefined\cO
\else
\newcommand\cO{\mathcal{O}}
\fi
\ifdefined\cP
\else
\newcommand\cP{\mathcal{P}}
\fi
\ifdefined\cQ
\else
\newcommand\cQ{\mathcal{Q}}
\fi
\ifdefined\cR
\else
\newcommand\cR{\mathcal{R}}
\fi
\ifdefined\cS
\else
\newcommand\cS{\mathcal{S}}
\fi
\ifdefined\cT
\else
\newcommand\cT{\mathcal{T}}
\fi
\ifdefined\cU
\else
\newcommand\cU{\mathcal{U}}
\fi
\ifdefined\cV
\else
\newcommand\cV{\mathcal{V}}
\fi
\ifdefined\cW
\else
\newcommand\cW{\mathcal{W}}
\fi
\ifdefined\cX
\else
\newcommand\cX{\mathcal{X}}
\fi
\ifdefined\cY
\else
\newcommand\cY{\mathcal{Y}}
\fi
\ifdefined\cZ
\else
\newcommand\cZ{\mathcal{Z}}
\fi
\ifdefined\bA
\else
\newcommand\bA{\mathbb{A}}
\fi
\ifdefined\bB
\else
\newcommand\bB{\mathbb{B}}
\fi
\ifdefined\bC
\else
\newcommand\bC{\mathbb{C}}
\fi
\ifdefined\bD
\else
\newcommand\bD{\mathbb{D}}
\fi
\ifdefined\bE
\else
\newcommand\bE{\mathbb{E}}
\fi
\ifdefined\bF
\else
\newcommand\bF{\mathbb{F}}
\fi
\ifdefined\bG
\else
\newcommand\bG{\mathbb{G}}
\fi
\ifdefined\bH
\else
\newcommand\bH{\mathbb{H}}
\fi
\ifdefined\bI
\else
\newcommand\bI{\mathbb{I}}
\fi
\ifdefined\bJ
\else
\newcommand\bJ{\mathbb{J}}
\fi
\ifdefined\bK
\else
\newcommand\bK{\mathbb{K}}
\fi
\ifdefined\bL
\else
\newcommand\bL{\mathbb{L}}
\fi
\ifdefined\bM
\else
\newcommand\bM{\mathbb{M}}
\fi
\ifdefined\bN
\else
\newcommand\bN{\mathbb{N}}
\fi
\ifdefined\bO
\else
\newcommand\bO{\mathbb{O}}
\fi
\ifdefined\bP
\else
\newcommand\bP{\mathbb{P}}
\fi
\ifdefined\bQ
\else
\newcommand\bQ{\mathbb{Q}}
\fi
\ifdefined\bR
\else
\newcommand\bR{\mathbb{R}}
\fi
\ifdefined\bS
\else
\newcommand\bS{\mathbb{S}}
\fi
\ifdefined\bT
\else
\newcommand\bT{\mathbb{T}}
\fi
\ifdefined\bU
\else
\newcommand\bU{\mathbb{U}}
\fi
\ifdefined\bV
\else
\newcommand\bV{\mathbb{V}}
\fi
\ifdefined\bW
\else
\newcommand\bW{\mathbb{W}}
\fi
\ifdefined\bX
\else
\newcommand\bX{\mathbb{X}}
\fi
\ifdefined\bY
\else
\newcommand\bY{\mathbb{Y}}
\fi
\ifdefined\bZ
\else
\newcommand\bZ{\mathbb{Z}}
\fi

\newcommand{\id}{\mathrm{id}}

\newcommand{\Set}{{\sf Set}}

\newcommand{\basety}{\mathbbl{o}}

\newcommand\lampa{\lambda\wp}
\newsavebox{\foobox}

\newcommand\fancya{{\text{\normalfont a}}}
\newcommand\lama{\lambda\fancya}

\newcommand\lin{\multimap}
\newcommand\tensor{\otimes}

\newcommand\interp[1]{{\llbracket #1 \rrbracket}}

\newcommand\Hom[3]{\left[#2,#3\right]_{#1}}

\newcommand\unit{\mathbf{I}}
\newcommand\ev{\mathrm{ev}}

\newcommand\op{\mathrm{op}}

\newcommand\lam{\lambda}

\newcommand\Str{\mathrm{Str}}

\newcommand\NF{\mathrm{NF}}

\newcommand\lamequiv{{\beta\eta}}

\newcommand\longto\longrightarrow

\newcommand\Int{\mathsf{Int}}

\newcommand{\ttreverse}{\mathtt{reverse}}

\newcommand\vareps\varepsilon

\usepackage{tikz-cd}
\newcommand{\catAutomata}[1]{\catname{Shape}_{#1}}
\newcommand{\catname}[1]{\textsf{#1}}

\newcommand{\catTransition}{\catname{TransDiag}}

\newcommand{\emptyString}{\epsilon}
\newcommand{\stringStart}{\triangleright}
\newcommand{\stringEnd}{\triangleleft}
\newcommand{\str}[1]{{#1}^{*}}

\newcommand{\forward}[1]{#1^\rightarrow}
\newcommand{\backward}[1]{#1^\leftarrow}
\newcommand{\ioSeparator}{/}
\usetikzlibrary {positioning, fit, backgrounds, automata}
\tikzset{
  plus/.style={circle,draw,fill=orange!80, node contents={$+$}},
  minus/.style={circle, draw, fill=cyan!80, node contents={$-$}},
  box/.style={fill=lightgray!50, node contents={}},
}
\newcommand{\sIn}{\mathrm{in}}
\newcommand{\sOut}{\mathrm{out}}
\newcommand{\sState}{\mathrm{states}}
\usepackage{formal-grammar}
\newcommand{\lolly}{\multimap}
\newcommand{\lI}{\unit}

\newcommand{\nullEnv}{\cdot}
\newcommand{\hasType}[2]{#1 : #2}
\newcommand{\judgement}[4]{#1; #2 \vdash \hasType{#3}{#4}}

\newcommand{\eqdef}{\mathrel{:=}}

\newcommand\CMRupdate[2]{\mathrm{RegUp}(#1,#2)}
\newcommand\registerContent{\mathrm{RegContent}}

\newcommand\Church[1]{\mathrm{Church}(#1)}

\newcommand{\dual}[1]{#1^{*}}
\newcommand{\leftdual}[1]{\dual{{}}\!#1}
\newcommand{\FinSet}{\mathsf{FinSet}}
\newcommand{\FinRel}{\mathsf{FinRel}}

\newcommand{\ELEMENTARY}{\textsc{Elementary}}

\usepackage{todonotes}

\def\conf{MFPS 2024}
\def\myconf{mfps40} 	%%Fill in the acronym for your conference (with year)
\volume{4}			%Fill in the ENTICS volume number here
\def\volu{4}			% and here
\def\papno{19}			%Fill in your paper number here
\def\mydoi{14804}%
\def\lastname{Pradic, Price}
\def\runtitle{Implicit automata in $\lambda$-calculi III} 

\def\copynames{C. Pradic, I. Price}
\def\copyname{}
\begin{document}
\begin{frontmatter}
  \title{Implicit Automata in $\lambda$-calculi III:\\[1ex] Affine Planar String-to-string Functions\thanksref{ALL}}
 \thanks[ALL]{We thank Aurore Alcolei, Lê Thành D\~ung Nguy{\~{ê}}n and Arno Pauly for discussions about this work.\\
  The authors acknowledge support by Swansea University and grant to Swansea University a non‐exclusive, irrevocable, sub‐licensable, worldwide license to make the accepted manuscript available on its institutional repository. 
  }   
  \author{Cécilia Pradic\thanksref{a}\thanksref{myemail}}
   \author{Ian Price\thanksref{a}\thanksref{coemail}}
   \address[a]{Department of Computer Science\\ Swansea University\\
    Wales}  							
   \thanks[myemail]{Email: \href{mailto:c.pradic@swansea.ac.uk} {\texttt{\normalshape
        c.pradic@swansea.ac.uk}}}
  \thanks[coemail]{Email:  \href{mailto:2274761@swansea.ac.uk} {\texttt{\normalshape
        2274761@swansea.ac.uk}}}
\begin{abstract}
We prove a characterization of first-order string-to-string
transduction via $\lambda$-terms typed in non-commutative affine logic
that compute with Church encoding, extending the analogous known
characterization of star-free languages. We show that every
first-order transduction can be computed by a $\lambda$-term using a known
Krohn-Rhodes-style decomposition lemma. The converse direction is
given by compiling $\lambda$-terms into two-way reversible planar
transducers. The soundness of this translation involves showing that
the transition functions of those transducers live in a monoidal
closed category of diagrams in which we can interpret purely affine
$\lambda$-terms. One challenge is that the unit of the tensor of the
category in question is not a terminal object. As a result, our interpretation
does not identify $\beta$-equivalent terms, but it does turn $\beta$-reductions
into inequalities in a poset-enrichment of the category of diagrams.
\end{abstract}
\begin{keyword}
non-commutative linear logic, transducers, $\lambda$-calculus, automata theory,
Church encodings
\end{keyword}

\end{frontmatter}

\section{Introduction}

The first author and Nguy{\~{ê}}n initiated a series of work that
compares the expressiveness of \emph{simply-typed affine $\lambda$-calculi} (in the sense of linear logic) and 
\emph{finite-state machine} from automata theory in~\cite{aperiodic}.
This endeavour is very much in the spirit of \emph{implicit computational complexity},
a field where one attempts to capture complexity-theoretic classes
of functions (rather than automata-theoretic) via various typed programming
languages, hence our borrowing of the term ``implicit''.

The starting point was to refine Hillebrand and Kanellakis' theorem~\cite[Theorem~3.4]{HillebrandKanellakis}
that states that the simply-typed $\lambda$-calculus captures regular languages
when computing over Church encodings.
Then, it was shown that one can also characterize \emph{star-free languages} via
the non-commutative affine $\lambda$-calculus ($\lampa$)~\cite{aperiodic}.
$\lampa$ features a function type that constrains arguments to be used at most
once and ``in order'', which restrains the available power.
It was conjectured that, when it comes to affine string-to-string functions,
$\lampa$ computes exactly \emph{first-order transductions} and its commutative
variant the larger class of \emph{regular transductions}~\cite{aperiodic}.
The latter was proven in follow-up work~\cite{freeadditives,titophd} and the main
contribution of this paper is to tackle the former, extending and
generalizing~\cite[Theorem 1.7]{aperiodic}.

\begin{restatable}{thm}{mainplanarstring}
\label{thm:main-planar-str}
Affine string-to-string $\lampa$-definable functions and first-order string transductions
coincide.
\end{restatable}

That every first-order transduction is $\lampa$-definable follows from
a decomposition lemma that states that all such transductions are compositions
of elementary building blocks that can be coded in $\lampa$. Most
of this coding was already done in~\cite[Theorem 4.1]{aperiodic}.
The more interesting direction is the converse, which is proven using a
semantic evaluation argument to compile $\lampa$-definitions into
\emph{two-way planar reversible finite transducers} (2PRFTs), a variant of two-way transducers
that were recently shown to capture exactly first-order transductions~\cite{planartrans}.
The semantics in question targets a non-symmetric monoidal-closed category
$\catTransition_\Gamma$ in which transitions of 2PRFTs find a natural home as
morphisms.

Much like other semantic evaluation arguments like Hillebrand and Kanellakis'
or in higher-order model checking~\cite{grellois,grelloismellies}, a nice aspect
is that automata will be computed in a very straightforward way from terms once
things are set up, and this computation will even be polynomial-time here provided
we are given a normal term as input. However, one difficulty we are going to run into will
have to do with the fact that our calculus is not linear but \emph{affine} and
that $\catTransition_\Gamma$ does not have a terminal object. We will
still manage to use it as an interpretation target for $\lampa$
by noticing that it carries a $\catname{Poset}$-enriched structure and
showing that this is enough to have an interpretation\footnote{We suspect this
can be characterized as an initiality theorem stating that there is a minimal oplax
strong monoidal-closed functor from initial affine monoidal-closed categories
to $\catname{Poset}$-with-$\bot$-enriched monoidal-closed categories, but we leave this characterization,
which would require dealing with tensors in the $\lampa$-calculus, for future work.} of terms $\interp{-}$ such
that $\interp{t_\NF} \le \interp{t}$ when $t$ evaluates to $t_\NF$ via $\beta$-reduction.

\paragraph{Plan of the paper} In Section~\ref{sec:background}, we review
the standard notions concerning $\catname{Poset}$-enriched categories and
the non-commutative $\lambda$-calculus we will require. We then explain
in Section~\ref{sec:fotranslambda} what it means for a string-to-string function
to be $\lampa$-definable and what 2PRFTs are. The latter we take as an opportunity
to introduce $\catTransition_\Gamma$ and define transitions of 2PRFTs
as morphisms in those categories. In Section~\ref{sec:main}, we prove
Theorem~\ref{thm:main-planar-str}. Finally, we conclude with some observations
concerning the commutative case and tree transductions that follow from our
work in Section~\ref{sec:main} before evoking some further research directions
that could most probably build on the material presented here.

\paragraph{Related work} For a more comprehensive overview of ``implicit automata
in $\lambda$-calculi'', one may consult the introductions of~\cite{freeadditives,titophd}.
Regarding this paper more specifically, the other most relevant works are the one
leading up to the introduction of 2PRFTs in~\cite{planartrans}, which mostly comes
from Hines' suggestion in~\cite{hinesplanar}, which itself drew on
Girard's geometry of interaction programme~\cite{towardsgoi} and Temperley-Lieb
algebras~\cite{AbramskyTemperleyLieb,PresentationsTL}.
We use categorical automata in the sense of
Colcombet and Petrişan~\cite{colcombetpetrisan} for practical purposes
similar to~\cite{freeadditives}. While categorical frameworks are used
to give generic results for various classes of automata by, e.g., viewing them
as algebras~\cite{adamek-trnkova-book,arbib-adjoint-machines,goguen-realization-cc}, as
coalgebras~\cite{rutten-coalg} or as dependent lenses~\cite{spivak2020poly},
here we simply use a categorical definition of 2PRFTs so that it may be easily
related to the semantics of the $\lampa$-calculus. In particular, we will
focus on the categories $\catTransition_\Gamma$ (for $\Gamma$ ranging over alphabets) and no other categories for most of the paper.
While we are not aware of a source that defines exactly $\catTransition_\Gamma$,
it is likely that close matches exist in the literature as it admits a straightforward
inductive presentation.
A similar construction is the operad of spliced
words in~\cite[Example 1.2]{MelliesZeilbergerContours}, where
the more general operad of spliced contours~\cite[Definition 1.1]{MelliesZeilbergerContours} is used
to analyze and generalize the Chomsky-Sch\"utzenberger representation theorem.

\section{Background}
\label{sec:background}

\subsection{Categorical preliminaries}
In the rest of this subsection, we list the key definitions related to
$\catname{Poset}$-enriched strict monoidal categories. In particular, we specialize
the definitions from general enriched category~\cite{basicsenrichedcat} to the
$\catname{Poset}$-enriched case for the convenience of the reader.

For notations, we use $\circ$ for composition, but also $;$ for composition
written in the reverse order ($f ; g = g \circ f$) when it is more convenient.
We write $\id_A$ for the identity at object $A$ and $\Hom{\cC}{A}{B}$ for the
set of morphisms of $\cC$ with domain $A$ and codomain $B$. When the ambient
category is clear from context or $\Set$, we sometimes write $f : A \to B$ to mean
that $f$ is a morphism from $A$ to $B$.

\begin{definition}
A \emph{category} $\cC$ is said to be \emph{$\catname{Poset}$-enriched} if it is
enriched in the category of posets and monotone functions, i.e., if
  for every objects $A$ and $B$, $\Hom{\cC}{A}{B}$ is a partially
  ordered set and composition $\Hom{\cC}{B}{C}\times\Hom{\cC}{A}{B}\to\Hom{\cC}{A}{C}$
  is monotone with respect to the product ordering on $\Hom{\cC}{B}{C}\times\Hom{\cC}{A}{B}$.

  A \emph{functor} $T : \cC \to \mathcal{D}$ between $\catname{Poset}$-enriched categories is
\emph{$\catname{Poset}$-enriched} if it is enriched in the
category of posets and monotone functions, i.e., for any objects $A$ and
  $B$ of $\cC$, $T_{A,B} : \Hom{\cC}{A}{B} \to \Hom{\mathcal{D}}{T(A)}{T(B)}$ is monotone.
A \emph{$\catname{Poset}$-enriched natural transformation} between
$\catname{Poset}$-enriched functors is just a natural transformation.
\end{definition}

\begin{definition}
A ($\catname{Poset}$-enriched) category $\cC$ is strict monoidal when
we have an (enriched) functor $\tensor : \cC^2 \to \cC$ and an object $\unit$
such that $(\tensor, \unit)$ and $(\tensor, \id_\unit)$ induce monoid structures
on the objects and morphisms of $\cC$.
\end{definition}

Note that we did not include a symmetry $A \tensor B \cong B \tensor A$ in our
definition of monoidal.
Although the coming definitions also make sense for non-strict monoidal categories,
throughout the rest of the paper, we will consider strict monoidal
categories only.

\begin{definition}
A ($\catname{Poset}$-enriched) monoidal category $\left(\cC, \lI, \tensor\right)$ is
\emph{closed} if for each object $X$ of $\cC$, the (enriched) functor
$(- \tensor X) : \cC \to \cC$ has an (enriched) right adjoint $(X \lolly -) : \cC \to \cC$, i.e.,
for any triple of objects $X, Y, Z$ we have a natural isomorphism
  $\Lambda_{X, Y, Z} : \Hom{\cC}{X \tensor Y}{Z} \cong \Hom{\cC}{X}{Y \lolly Z}$ which
is monotone.
  We will write $\ev_{Y, Z}$ for the counit of the adjunction\footnote{It is equal to $\Lambda_{Y,Y,Z}^{-1}(\id_{Y \lolly Z}) : (Y \lolly Z) \tensor Y \to
  Z$ by definition and corresponds to an evaluation morphism $(Y \lolly Z) \tensor Y \to Z$ used to interpret function application.}.
\end{definition}

As we are also interested in categories with a dualising structure, it
would be natural to ask for an (enriched) compact-closed category.
However, to the author's knowledge, there is no clear consensus on the
``correct'' definition of compact-closed category when the tensor is
not symmetric. One such candidate, a restricted version of pivotal
category, was put forward by Freyd \& Yetter~\cite{freydyetter89} and
is appropriate to our needs. The following definitions come from
Selinger's survey of graphical languages~\cite{selingergraphical}.

\begin{definition}
  In a monoidal category, an \emph{exact pairing} between two objects $A$ and
  $B$, is given by a pair of maps $\eta : \unit \to B \otimes A, \varepsilon : A \otimes B \to \unit$,
  called respectively \emph{cups} and \emph{caps},
  such that the following two triangles commute\footnote{These
  equations are typically called the ``yanking'' or ``zigzag'' equations.}:
  \[
  \begin{tikzcd}
    A \arrow[r, "\id_A\otimes\eta"] \arrow{dr}[swap]{\id_A} & A \otimes B \otimes A \arrow[d, "\varepsilon\otimes\id_A"]\\
    & A
  \end{tikzcd}
  \qquad
  \begin{tikzcd}
    B \arrow[r, "\eta\otimes\id_{B}"] \arrow{dr}[swap]{\id_{B}} & B \otimes A \otimes B \arrow[d, "\id_{B}\otimes\varepsilon"]\\
    & B
  \end{tikzcd}
  \]

  In an exact pairing, $B$ is called the \emph{right dual} of $A$ and
  A is called the \emph{left dual of B}.
\end{definition}

\begin{definition}
A monoidal category is \emph{left (resp. right) autonomous} if every object $A$ has
a left (resp. right) dual, which we denote $\leftdual{A}$ (resp $\dual{A}$). It is \emph{autonomous} if it is both left and right autonomous.
\end{definition}

Any choice of duals $\dual{A}$ and cups and caps $\varepsilon_A$, $\eta_A$ for
every object $A$ in a left autonomous category $\cC$ extends
$(-)^*$ to a functor $\cC \to \cC^\op$ by setting
$f^* = (\eta_A \tensor \id_{\dual{B}}); (\id_{\dual{A}} \tensor f \tensor \id_{\dual{B}}); (\id_{\dual{A}} \tensor \varepsilon_B)$
when $f : A \to B$. We then also have that the chosen
cups and caps are natural transformations. Similar definitions can be
made for right autonomous categories.

\begin{definition}
  \label{def:pivotal}
A \emph{pivotal category} is a right autonomous category equipped with
a monoidal natural transformation $i_A : A \to \dual{\dual{A}{}}$.
We are primarily interested in the case where $i_A$ is the identity,
in which case, we refer to it as a \emph{strict} pivotal category.
\end{definition}

The following lemma shows that pivotal categories allow us treat left
and right duals as the same and define closure in terms of duals.

\begin{lemma}
  \label{lem:pivotal-closed}
Pivotal categories are autonomous and closed.
\end{lemma}
\begin{proof}
Since $\dual{{\dual{A}}{}} \cong A$ and $\dual{{\dual{A}}{}}$ is the right
dual of $\dual{A}$, it follows that $\dual{A}$ is also left dual of $A$.

To show monoidal closure, define the functor $(B \lolly -) \eqdef (- \tensor \dual{B})$.
We can construct the adjunction by setting
$\Lambda_{A,B,C}(f) = (\id_A \tensor \eta_B) ; (f \tensor \id_{B^*})$, which
has inverse $\Lambda_{A,B,C}^{-1}(g) = (g \tensor \id_B) ; (\id_C \tensor \varepsilon_B)$.
That $\Lambda$ and $\Lambda^{-1}$ are inverse is provable thanks to the yanking
equations.
\end{proof}

\subsection{The planar $\lambda$-calculus $\lampa$}

    For most of the paper, we will be working in the non-commutative fragment of the affine $\lambda$-calculus
    that we call $\lampa$. \emph{Types} of $\lampa$, that we typically write with
    the greek letter $\tau, \sigma$ and $\kappa$, are inductively generated by a designated base type $\basety$
    and two type constructors $\lin$ and $\to$ corresponding respectively to
    \emph{affine} and unrestricted function types.
    We will have the following restrictions for the function spaces built with $\multimap$:
    \begin{itemize}
      \item arguments must be used at most once\\
        \phantom{a} \hfill  {\footnotesize ($\lambda f. \lambda x. \; f \; (f \; x)$ does not have type $(\basety \lin \basety) \lin \basety \lin \basety$)}
      \item arguments must occur in order in application.\\
        \phantom{a} \hfill  {(\footnotesize $\lambda x. \lambda f. \; f \; x$ does not have type $\basety \lin (\basety \lin \basety) \lin \basety$)}
    \end{itemize}
    We introduce both the syntax and the typing rules of $\lampa$
    (which, in particular, enforce those restrictions) in Figure~\ref{fig:lamp}.
    Throughout, we formally need to manipulate terms that come with their type
    derivations rather than raw terms, but we will often simply write out
    terms rather than typing judgement for legibility. 
    We call the fragment where types do not contain the non-affine arrow $\to$ \emph{purely affine}.

    To make those term compute, we define the capture-avoiding substitution of $x$ by a term $u$ in $t$ by $t[u/x]$
    as usual, as well as the relation $\to_\beta$ of \emph{$\beta$-reduction} as being
    the least relation satisfying
    $(\lambda x. \; t) \; u \to_\beta t[u/x]$ for all well-typed expressions (of the same type)
    and being closed by congruence. Call $\to_\beta^*$ its reflexive transitive closure.
    An expression of shape $(\lambda x. \; t) \; u$ is
    called a \emph{$\beta$-redex} and a term containing no such redex is called
    \emph{normal}.
    The least congruence
    containing all clauses $t =_\eta \lambda x. \; t \; x$ for every $t$ with no
    occurrence of $x$ which has a function type is called $\eta$-equivalence.
    Two terms are called $\beta\eta$-equivalent if they can be
    related by the least equivalence relation containing $\to_\beta$ and $=_\eta$.
    We write $=_{\beta\eta}$ for $\beta\eta$-equivalence.
    
    Every rewriting sequence involving $\to_\beta$ and well-typed terms terminates.
    \begin{proposition}[standard argument, see also {\cite[Proposition 2.3]{aperiodic}}]
      \label{prop:normalization}
      For every $\Psi; \; \Delta \vdash t : \tau$, there is a normal term
      $t_\NF$ with the same typing
      such that $t \to_\beta^* t_\NF$.
    \end{proposition}

\begin{figure}
\begin{center}
\begin{mathpar}
\inferrule*{ }
            {\judgement{\Psi,\hasType{x}{\tau}, \Psi'}{\Delta}{x}{\tau}}
\and
\inferrule*{ }
            {\judgement{\Psi}{\Delta, \hasType{x}{\tau}, \Delta'}{x}{\tau}}
\\
\inferrule*{\judgement{\Psi}{\Delta, \hasType{x}{\tau}}{t}{\sigma}}
           {\judgement{\Psi}{\Delta}{\lambda{}x.t}{\tau\lolly\sigma}}
\and
\inferrule*{\judgement{\Psi}{\Delta}{t}{\tau\lolly\sigma} \and
            \judgement{\Psi}{\Delta^{\prime}}{u}{\tau}}
            {\judgement{\Psi}{\Delta, \Delta^{\prime}}{t \; u}{\sigma}}
\\
\inferrule* {\judgement{\Psi, \hasType{x}{\tau}}{\Delta}{t}{\sigma}}
            {\judgement{\Psi}{\Delta}{\lambda x.t}{\tau\to\sigma}}
\and
\inferrule* {\judgement{\Psi}{\Delta}{t}{\tau\to\sigma}
                  \\ \judgement{\Psi}{\nullEnv}{u}{\tau}}
                 {\judgement{\Psi}{\Delta}{t \; u}{\sigma}}

\end{mathpar}
  \end{center}
  \caption{Syntax and typing rules for $\lampa$. The contexts $\Psi$ and $\Delta$ are
  lists of pairs $x : \tau$ containing a variable name $x$ and some type $\tau$.
  We assume that all variables appearing in a context and under binders are pairwise distinct and that terms and derivations are defined up to $\alpha$-renaming.}
\label{fig:lamp}
\end{figure}

\section{First-order string-to-string transductions in the planar affine $\lambda$-calculus}
\label{sec:fotranslambda}

\subsection{Definable string-to-string functions in the planar affine $\lambda$-calculus}

In order to discuss string functions in $\lampa$, we need to discuss
how they are encoded. For that, we use the same framework as
in~\cite{moreau23syn,titophd}.
In the pure (i.e.\ untyped) $\lambda$-calculus and its polymorphic typed variants
such as System F, the canonical way to encode
inductive types is via \emph{Church encodings}.
Such encodings are typable in the simply-typed $\lambda$-calculus by dropping
the prenex universal quantification that comes with them in polymorphic calculi.
For instance, for natural numbers and strings over $\{a,b\}$, writing $\Church{w}$ for the
Church encoding of $w$, we have
$
\Church{aab} =
  \lambda a. \lambda b. \lambda \epsilon. \; \underline{aab} = 
  \lambda a. \lambda b. \lambda \epsilon. \; a \; (a \; (b \; \epsilon))
$.

Conversely, a consequence of normalization is that any closed simply typed
$\lambda$-term ``of type string'' is $\lamequiv$-equivalent to the Church
encoding of some string.
In the rest of this paper, we use a type for Church encodings of strings
that is finer than usual and not expressible without $\lin$, first introduced
in~\cite[\S5.3.3]{girardLL}.

\begin{definition}
  Let $\Sigma$ be an alphabet. We define $\Str_\Sigma$ as
  $\underbrace{(\basety \lin
  \basety) \to \ldots \to (\basety \lin \basety)}_{\text{$|\Sigma|$ times}} \to \basety \to \basety$.
\end{definition}

\begin{definition}
Given an alphabet $\Sigma = \{a_1, \ldots, a_n\}$, define the
signature $\underline{\Sigma}$ as $a_1 : \basety \lin \basety, \ldots, a_n : \basety \lin \basety, \epsilon : \basety$.
For every word $w \in \Sigma^*$ define the typed term $\underline{\Sigma}; \; \cdot \vdash \underline{w} : \basety$
and the closed term $\Church{w} : \Str_\Sigma$ by
\[
\underline{\epsilon} = \epsilon
\qquad
  \underline{a_iw'} = a_i \; \underline{w'}
  \qquad \text{and} \qquad
\Church{w} = \lam a_1 \ldots \lam a_n. \lam \epsilon. \; \underline{w}\]
\end{definition}

We can then show the following by inspecting the normal form and using Proposition~\ref{prop:normalization}.

\begin{lemma}
For every $\underline{\Sigma}; \cdot \vdash t : \basety$, $t$ is $\beta\eta$-equivalent
to a unique $\underline{w_t}$ and, a fortiori,
for every $\underline{\Sigma}; \cdot \vdash u : \basety$, $u$ is
  $\beta\eta$-equivalent to a unique $\Church{w_u}$.
\end{lemma}

As a consequence, $\lampa$-terms of type $\Str_\Sigma \to \Str_\Gamma$
correspond to functions $\Sigma^* \to \Gamma^*$, but have a limited expressivity.
We consider a natural
extension of these by allowing to emulate a limited kind of polymorphism via
\emph{type substitutions} $\tau[\kappa]$ defined as follows.
\[
\basety[\kappa] = \kappa \qquad \text{and} \qquad (\tau \lin \sigma)[\kappa] = \tau[\kappa] \lin \sigma[\kappa]
\]

Type substitutions extend in the obvious way to typing contexts, and even
to \emph{typing derivations}, so that 
$\Psi; \; \Delta \vdash t : \tau$ entails $\Psi[\kappa]; \; \Delta[\kappa] \vdash t : \tau[\kappa]$.
In particular, it means that a Church encoding $t : \Str_{\bf \Sigma}$ is
also of type $\Str_{\bf \Sigma}[\kappa]$ for any type $\kappa$.
This ensures that the following notion of definable string-to-string functions
makes sense and is closed under function composition.

\begin{definition}
\label{def:lam-definable}
A function $f : \Sigma^* \to \Gamma^*$ is called \emph{affine $\lampa$-definable} when
there exists a \emph{purely affine} type $\kappa$ together with a $\lambda$-term
$\mathtt{f} : \Str_{\Sigma}[\kappa] \lin \Str_{\Gamma}$
such that $f$ and $\mathtt{f}$ coincide up to Church encoding; i.e.,
for every string $t \in \Sigma^*$, $\Church{f(t)} =_{\lamequiv} \mathtt{f} \; \Church{t}$.

\end{definition}

\begin{example}
\label{ex:lam-rev}
  The function $\texttt{reverse} : \Sigma^* \to \Sigma^*$ that reverses its input
is affine $\lampa$-definable. Supposing that we have $\Sigma = \{a_1, \dots, a_k\}$,
one $\lampa$-term that implements it is
\[\lam s. \lam a_1. \ldots \lam a_k. \lam \epsilon.\; s \; (\lam x.
  \lam z. x \; (a_1 \; z)) \ldots (\lam x. \; (a_k \; z)) \; (\lam x.x) \;
  \epsilon
  \;:\; \Str_\Sigma[\basety \lin \basety] \lin \Str_\Sigma \]
\end{example}

This definition involves terms defined in the full calculus that still requires
to work with the $\to$ type constructor that occurs in $\Str$.
But we also have an equivalent characterization in terms of purely affine terms.
This characterization is obtained by inspecting the
normal form of a $\lampa$ definition.

\begin{lemma}[{particular case of~\cite[Lemma 5.25]{titophd}, easier to prove from Proposition~\ref{prop:normalization}}]
\label{lem:laml-niceshape-str}
Let $\Sigma = \{a_1, \ldots, a_n\}$
and $\Gamma = \{ b_1, \ldots, b_k\}$ be alphabets.
Up to $\lamequiv$-equivalence, every term of type $\Str_{\Sigma}[\kappa] \lin \Str_{\Gamma}$ is of the shape $\lam s. \lam b_1. \ldots \lam b_k.
  \lam \epsilon. \; o \; (s \; d_1 \; \ldots \; d_n \; d_\epsilon)$
such that $o$, $d_\epsilon$ and the $d_i$s are purely linear $\lampa$-terms with no occurrence of $s$, that is, terms such as we
have typing derivations
\[
\underline{\Gamma} ; \cdot \vdash o : \kappa \lin \basety \qquad
\underline{\Gamma} ; \cdot \vdash d_i : \kappa \lin \kappa \qquad
\underline{\Gamma} ; \cdot \vdash d_\epsilon : \kappa
\]
\end{lemma}
This lemma and the fact that $\mathtt{reverse}$ is definable mean that an
affine $\lampa$-definable function
$\Sigma^* \to \Gamma^*$ can, without loss of generality, be given by
a \emph{$\lampa$-transducer}, which we define as follows (see e.g.~\cite[Definition 2.6]{nguyenvanoni} or~\cite[Definition 3.22]{freeadditives} for similar
definitions).

\begin{definition}
  \label{def:normalAffineStr2Str}
A \emph{$\lampa$-transducer} with input $\Sigma^*$ and output $\Gamma^*$ is given
by the
following types and terms from the purely affine planar $\lambda$-calculus
with constants in $\underline{\Gamma}$:
\begin{itemize}
  \item an iteration type $\kappa$,
  \item for each $a \in \Sigma$, a term $d_a : \kappa \lin \kappa$ over the signature $\underline{\Gamma}$,
  \item a term $d_\epsilon : \kappa$
  \item and a term $o : \kappa \lin \basety$.
\end{itemize}
The underlying function is then defined by mapping a word $w_0\ldots w_n$
to the word corresponding to the normal form of
  $o \; (d_{w_n} \; (\ldots (d_{w_0} \; d_\epsilon) \ldots))$.
\end{definition}

\subsection{The category of planar diagrams}

We will now introduce a category of what we are going to call 
\emph{planar diagrams}. The idea is that the morphisms may be represented by
graphs with (an ordered set of) vertices labelled by polarities $p \in \{-, +\}$ and edges labelled
by words over some fixed output alphabet $\Gamma$. Also given would be a
partition of the vertices into input and outputs, and then the composition
would be represented by pasting the diagrams together and concatenating labels,
in an order prescribed by the polarities and whether the nodes involved are
inputs or outputs. One such diagram is pictured in Figure~\ref{fig:morphism}.

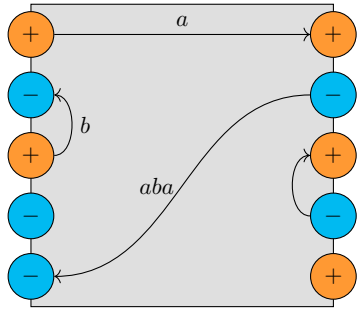
\begin{figure}[h]
  \begin{center}
    \begin{tikzpicture}
      [on grid,->, scale = 0.8, transform shape, baseline=(current bounding box.center)]

      \node (in 1) at (0,0) [plus];
      \node (in 2) [below=of in 1] [minus];
      \node (in 3) [below=of in 2] [plus];
      \node (in 4) [below=of in 3]  [minus];
      \node (in 5) [below=of in 4]  [minus];

      \node (out 1) [right=5 of in 1][plus];
      \node (out 2) [below=of out 1] [minus];
      \node (out 3) [below=of out 2] [plus];
      \node (out 4) [below=of out 3] [minus];
      \node (out 5) [below=of out 4] [plus];

      \draw (in 1) to node [above] {$a$} (out 1);
      \draw (in 3) to [out=0, in=0] node [right] {$b$} (in 2);
      \draw (out 4) to [bend right=270] (out 3);
      \draw (out 2) to [out=180, in=0] node [left] {$aba$} (in 5);

      \begin{scope}[on background layer]
        \draw[box] ([yshift=.5cm]in 1) rectangle ([yshift=-.5cm]out 5);
      \end{scope}
    \end{tikzpicture}
  \end{center}
  \caption{A geometric realization of a morphism from $+-+--$ to $+-+-+$.
  The edge directions are not part of the definition, but inferred from the
  polarity labels of the source and targets. When the label is $\emptyString$,
  we omit it from the picture.}
  \label{fig:morphism}
\end{figure}

The major restriction that we will put on the diagrams living on our category is that
they be \emph{planar}. While we will define these morphisms in a combinatorial
way for simplicity, this condition is more intuitive when interpreted geometrically.
A geometric interpretation of a diagram can be given by
writing out the nodes in order on the boundary of a bounding rectangle (filled in grey in
Figure~\ref{fig:morphism}), the inputs sitting on the left boundary and outputs
on the right boundary, and tracing out the edges within that square. A diagram
is then \emph{geometrically planar} when it is possible to do so without making the
edges cross.

On the other hand, the combinatorial definition goes as follows.

\begin{definition}
\label{def:planar-combinatorial}
$(V, <, E)$, consisting of an undirected graph $(V, E)$ ($E \subseteq [V]^2$)
and a total order $<$ over $V$, is called \emph{combinatorially planar} if for
every four vertices $a < b < c < d$ then we do not have
both edges between $a$ and $c$ and between $b$ and $d$.
\end{definition}

Checking that a combinatorial planar structure can be realized as a geometrical
planar structure is relatively straightforward. Proving that conversely 
a structure with a geometrically planar realization is combinatorially planar
can be done using the Jordan curve theorem.

While the diagrams formally do not have a direction, an intended traversal
direction is going to be induced by the label of the vertices and whether
they are in the input or output sets. More precisely
\begin{itemize}
\item if $v$ is an input vertex of polarity $+$ or an output vertex of
polarity $-$, then it is an implicit \emph{source} and
\item if $v$ is an output vertex of polarity $+$ or an input vertex of polarity
$-$, then it is an implicit \emph{target}.
\end{itemize}
In morphisms, we will restrict edges so that they contain exactly one implicit source
and target, so overall they are all orientable. This allows to define the
composition $f \circ g$ of two diagrams unambiguously.
This can be done for geometrical representations of $f$ and $g$ as follows:
\begin{enumerate}
\item paste the two diagrams together, identifying the output boundary of $g$
with the input boundary of $f$
\item take the new bounding rectangle to be the union of those for $f$ and $g$;
erase the nodes that do not belong to its boundary, as well as the edges that
dangle in its interior and loops
\item concatenate the labels along the implicit direction of the edges they are
labelling
\end{enumerate}
The way we restricted the edges so that they may be oriented makes sure that
the last step is well-defined and yields a picture where each interior edge
is unambiguously labelled by a word. This process, pictured on an example in 
Figure~\ref{fig:composition}, can be easily adapted beat-for-beat with the
combinatorial definition.
However, checking that this yields a diagram which is still planar is more
easily done geometrically.

\begin{figure}[h]
  \begin{center}
    \begin{tikzpicture}[->,scale=0.7, every node/.style={transform
          shape}, on grid, baseline=(current bounding box.center)]

      \node (in 1) at (0,0) [plus];
      \node (in 2) [below=of in 1] [minus];
      \node (in 3) [below=of in 2] [plus];

      \node (out 1) [right=2 of in 1][plus];
      \node (out 2) [below=of out 1] [plus];
      \node (out 3) [below=of out 2] [minus];
      \node (out 4) [below=of out 3] [plus];
      \node (out 5) [below=of out 4] [minus];

      \draw (in 1) to node [above] {$a$} (out 1);
      \draw (out 3) to [out=180, in=180] node [left] {$b$} (out 2);
      \draw (in 3) to [out=0, in=180] node [left] {$c$} (out  4);

      \begin{scope}[on background layer]
        \draw[box] ([yshift=.5cm]in 1) rectangle ([yshift=-.5cm]out 5);
      \end{scope}
    \end{tikzpicture}
    ;
    \begin{tikzpicture}[->,scale=0.7, every node/.style={transform
          shape}, on grid, baseline=(current bounding box.center)]

      \node (in 1) at (0,0) [plus];
      \node (in 2) [below=of in 1] [plus];
      \node (in 3) [below=of in 2] [minus];
      \node (in 4) [below=of in 3]  [plus];
      \node (in 5) [below=of in 4]  [minus];

      \node (out 1) [right=2 of in 1][plus];
      \node (out 2) [below=of out 1] [plus];
      \node (out 3) [below=of out 2] [minus];

      \draw (in 1) to node [above] {$x$} (out 1);
      \draw (in 2) to (out 2);
      \draw (in 4) to [out=0, in=0] node [left] {$z$} (in 3);
      \draw (out 3) to [out=180, in=0] node [right] {$y$} (in 5);

      \begin{scope}[on background layer]
        \draw[box] ([yshift=.5cm]out 1) rectangle ([yshift=-.5cm]in 5);
      \end{scope}
    \end{tikzpicture}
    $\qquad\longmapsto\qquad$
    \begin{tikzpicture}[->,scale=0.7, every node/.style={transform shape},on grid, baseline=(current bounding box.center)]
      \node (in 1) at (0,0) [plus];
      \node (in 2) [below=of in 1] [minus];
      \node (in 3) [below=of in 2] [plus];

      \node (mid 1) [right=2 of in 1][plus];
      \node (mid 2) [below=of mid 1] [plus];
      \node (mid 3) [below=of mid 2] [minus];
      \node (mid 4) [below=of mid 3] [plus];
      \node (mid 5) [below=of mid 4] [minus];

      \draw (in 1) to node [above] {$a$} (mid 1);
      \draw (mid 3) to [out=180, in=180] node [left] {$b$} (mid 2);
      \draw (in 3) to [out=0, in=180] node [left] {$c$} (mid  4);

      \node (out 1) [right=2 of mid 1][plus];
      \node (out 2) [below=of out 1] [plus];
      \node (out 3) [below=of out 2] [minus];

      \draw (mid 1) to node [above] {$x$}  (out 1);
      \draw (mid 2) to (out 2);
      \draw (mid 4) to [out=0, in=0] node [left] {$z$} (mid 3);
      \draw (out 3) to [out=180, in=0] node [right] {$y$} (mid 5);

      \begin{scope}[on background layer]
        \draw[box] ([yshift=.5cm]in 1) rectangle ([yshift=-.5cm]mid 5);
        \draw[box] ([yshift=.5cm]out 1) rectangle ([yshift=-.5cm]mid 5);
      \end{scope}
    \end{tikzpicture}
    $\qquad\longmapsto\qquad$
    \begin{tikzpicture}[->,scale=0.7, every node/.style={transform shape}, on grid, baseline=(current bounding box.center)]

      \node (in 1) at (0,0) [plus];
      \node (in 2) [below=of in 1] [minus];
      \node (in 3) [below=of in 2] [plus];

      \node (out 1) [right=2 of in 1][plus];
      \node (out 2) [below=of out 1] [plus];
      \node (out 3) [below=of out 2] [minus];

      \draw (in 1) to node [above] {$ax$} (out 1);
      \draw (in 3) to [out=0, in=180] node [left] {$czb$} (out 2);

      \begin{scope}[on background layer]
        \draw[box] ([yshift=.5cm]in 1) rectangle ([yshift=-.5cm]out 3);
      \end{scope}

    \end{tikzpicture}
  \end{center}
  \caption{How morphisms compose}
  \label{fig:composition}
\end{figure}
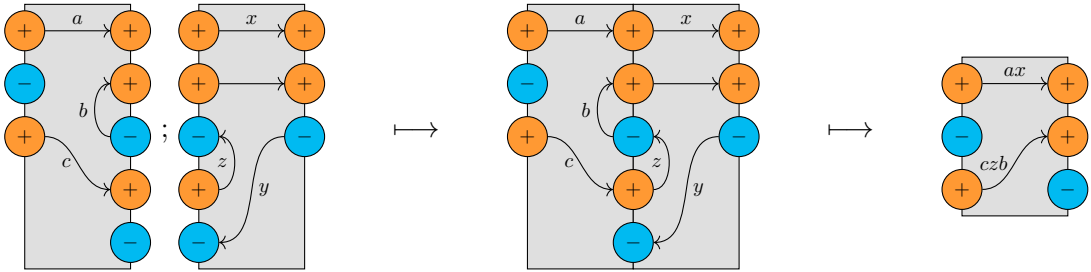

Let us summarize what is a legal diagram from a combinatorial standpoint.

\begin{definition}
  \label{def:diagram-combinatorial}
  A \emph{combinatorial planar diagram labelled by a monoid $M$} is a tuple
  $(V_{in}, V_{out}, \rho, <, E, \ell)$
  where
  \begin{itemize}[%leftmargin=3em,
        topsep=0.7em]
    \item $V_{in}$ and $V_{out}$ are disjoint finite sets of vertices
    \item $<$ is a total order over $V_{in} \cup V_{out}$
    \item $\rho : V_{in} \cup V_{out} \to \{+,-\}$ assigns polarities to vertices
    \item $E$ contains subsets of $V_{in} \cup V_{out}$ of size exactly two
    \item $\ell : E \to M$ assigns labels to edges
  \end{itemize}
  subject to the following restrictions, setting $V = V_{in} \cup V_{out}$:
  \begin{itemize}[%leftmargin=3em,
        topsep=0.7em]
    \item all vertices in $(V,E)$ must have degree at most one
    \item $v_{in} < v_{out}$ for every $v_{in} \in V_{in}$ and $v_{out} \in V_{out}$
    \item $(V, <, E)$ must be planar
    \item every edge $e \in E$ contains an implicit source as well as a target.
  \end{itemize}
\end{definition}

We can now give an official formal definition of categories of diagrams
$\catTransition_\Sigma$ where $\Sigma$ is going to be the output alphabet.
To make the monoidal structure on $\catTransition_\Sigma$
strict and our lives easier,
we will take objects to be words over $\{+,-\}$ rather than labelled
sets of inputs and outputs, and determine the vertices of the diagrams by
positions in the input and output objects.

\begin{definition}
\label{def:diagram-cat}
Let $\Sigma$ be a finite alphabet. The
category of planar diagrams over $\Sigma$, $\catTransition_\Sigma$, is defined
as follows.
\begin{itemize}
\item \textbf{Objects} are finite words in $\str{\{+,-\}}$.
\item \textbf{Morphisms}, for $A = a_1 \dots a_n$ and $B = b_1 \dots b_m$
  a morphism $A \to B$ is a planar combinatorial diagram labelled by $\Sigma^*$
  where:
\begin{itemize}
\item $V_{in} = \{(-1, 1), \dots, (-1, n)\}$
\item $V_{out} = \{(1, 1), \dots, (1, m)\}$
\item $<$ is defined by setting $(i, q) < (j, r)$ if and only if $(i, iq) <_{\mathrm{lex}} (j, jr)$ in
  the lexicographic order
\end{itemize}
\item \textbf{Identities} are given by diagrams where all labels are $\emptyString$
  and containing all possible edges $\{(-1,k),(1,k)\}$
\item \textbf{Composition} $h = f ; g$ is given by identifying the output
vertices $(1,k)$ of $g$ with the input vertices $(-1, k)$ of $f$ and composing
the combinatorial diagrams as explained above.
\end{itemize}
\end{definition}

The free monoid structure on objects $\str{\{+,-\}}$ extends to a
strict monoidal structure on $\catTransition_\Sigma$, i.e., tensoring of
objects is concatenation and the unit $\unit$ is $\emptyString$.
Over morphisms, tensoring can be pictured as putting two diagrams on top of
each other as in Figure~\ref{fig:monoidal-product}.
Note that the planarity of our diagrams means that this tensor cannot be equipped
with a symmetric structure and that $\unit$ is not a terminal object.

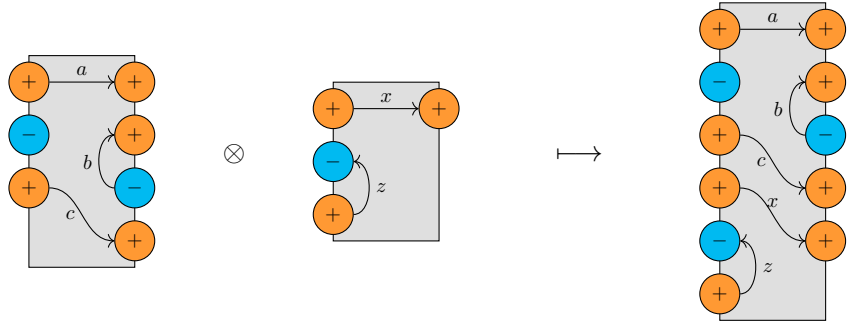
\begin{figure}[h]
  \begin{center}
    \begin{tikzpicture}[->,scale=0.7, every node/.style={transform
          shape}, on grid, baseline=(current bounding box.center)]

      \node (in 1) at (0,0) [plus];
      \node (in 2) [below=of in 1] [minus];
      \node (in 3) [below=of in 2] [plus];

      \node (out 1) [right=2 of in 1][plus];
      \node (out 2) [below=of out 1] [plus];
      \node (out 3) [below=of out 2] [minus];
      \node (out 4) [below=of out 3] [plus];

      \draw (in 1) to node [above] {$a$} (out 1);
      \draw (out 3) to [out=180, in=180] node [left] {$b$} (out 2);
      \draw (in 3) to [out=0, in=180] node [left] {$c$} (out  4);

      \begin{scope}[on background layer]
        \draw[box] ([yshift=.5cm]in 1) rectangle ([yshift=-.5cm]out 4);
      \end{scope}
    \end{tikzpicture}
    $\qquad\tensor\qquad$
    \begin{tikzpicture}[->,scale=0.7, every node/.style={transform
          shape}, on grid, baseline=(current bounding box.center)]

      \node (in 1) at (0,0) [plus];
      \node (in 2) [below=of in 1] [minus];
      \node (in 3) [below=of in 2]  [plus];

      \node (out 1) [right=2 of in 1][plus];

      \draw (in 1) to node [above] {$x$} (out 1);
      \draw (in 3) to [out=0, in=0] node [right] {$z$} (in 2);

      \begin{scope}[on background layer]
        \draw[box] ([yshift=.5cm]out 1) rectangle ([yshift=-.5cm]in 3);
      \end{scope}
    \end{tikzpicture}
    $\qquad\quad\longmapsto\quad\qquad$
    \begin{tikzpicture}[->,scale=0.7, every node/.style={transform shape}, on grid, baseline=(current bounding box.center)]

      \node (in 1) at (0,0) [plus];
      \node (in 2) [below=of in 1] [minus];
      \node (in 3) [below=of in 2] [plus];
      \node (in 4) [below=of in 3] [plus];
      \node (in 5) [below=of in 4] [minus];
      \node (in 6) [below=of in 5] [plus];

      \node (out 1) [right=2 of in 1][plus];
      \node (out 2) [below=of out 1] [plus];
      \node (out 3) [below=of out 2] [minus];
      \node (out 4) [below=of out 3] [plus];
      \node (out 5) [below=of out 4] [plus];

      \draw (in 1) to node [above] {$a$} (out 1);
      \draw (out 3) to [out=180, in=180] node [left] {$b$} (out 2);
      \draw (in 3) to [out=0, in=180] node [left] {$c$} (out  4);

      \draw (in 4) to [out=0, in=180] node [above] {$x$} (out 5);
      \draw (in 6) to [out=0, in=0] node [right] {$z$} (in 5);

      \begin{scope}[on background layer]
        \draw[box] ([yshift=.5cm]out 1) rectangle ([yshift=-.5cm]in 6);
      \end{scope}

    \end{tikzpicture}
  \end{center}
  \caption{The monoidal product of two morphisms}
  \label{fig:monoidal-product}
\end{figure}

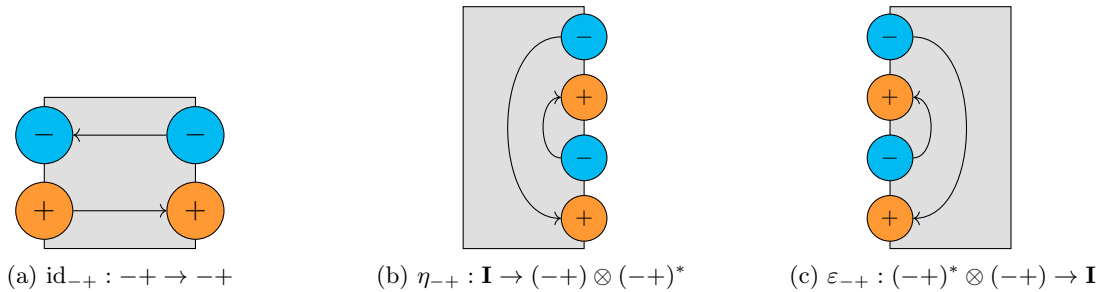
\begin{figure}[h]
  \centering
  \begin{subfigure}{0.3\linewidth}
    \centering
    \begin{tikzpicture}
      [on grid,->]

      \node (in 1) at (0,0) [minus];
      \node (in 2) [below=of in 1] [plus];

      \node (out 1) [right=2 of in 1][minus];
      \node (out 2) [below=of out 1] [plus];

      \draw (in 2) to (out 2);
      \draw (out 1) to (in 1);

      \begin{scope}[on background layer]
        \draw[box] ([yshift=.5cm]in 1) rectangle ([yshift=-.5cm]out 2);
      \end{scope}
    \end{tikzpicture}
    \caption{$\id_{-+} : -+ \to -+$}
  \end{subfigure}
  \begin{subfigure}{0.3\linewidth}
    \centering
    \begin{tikzpicture}
      [on grid,->, scale = 0.8, transform shape, baseline=(current bounding box.center)]
      \node (in first) at (0,0) [draw opacity=0] {};
      \node (in last) [below=3 of in first] [draw opacity=0] {};

      \node (out 1) [right=2 of in first][minus];
      \node (out 2) [below=of out 1] [plus];
      \node (out 3) [below=of out 2] [minus];
      \node (out 4) [below=of out 3] [plus];

      \draw (out 1) to [bend left=270] (out 4);
      \draw (out 3) to [bend left=90] (out 2);

      \begin{scope}[on background layer]
        \draw[box] ([yshift=.5cm]in 1) rectangle ([yshift=-.5cm]out 4);
      \end{scope}
    \end{tikzpicture}
    \caption{$\eta_{-+} : \lI \to (-+) \tensor \dual{(-+)}$}
  \end{subfigure}
  \begin{subfigure}{0.3\linewidth}
    \centering
    \begin{tikzpicture}
      [on grid,->, scale = 0.8, transform shape, baseline=(current bounding box.center)]
      \node (in 1) at (0, 0) [minus];
      \node (in 2) [below=of in 1] [plus];
      \node (in 3) [below=of in 2] [minus];
      \node (in 4) [below=of in 3] [plus];

      \node (out 1) [right=2 of in first] [draw opacity=0] {};
      \node (out 2) [below=3 of out 1] [draw opacity=0] {};

      \draw (in 1) to [bend right=270] (in 4);
      \draw (in 3) to [bend right=90] (in 2);

      \begin{scope}[on background layer]
        \draw[box] ([yshift=.5cm]in 1) rectangle ([yshift=-.5cm]out 2);
      \end{scope}
    \end{tikzpicture}
    \caption{$\varepsilon_{-+} : \dual{(-+)} \tensor (-+) \to \lI$}
  \end{subfigure}
  \caption{Identity, cup and cap for the object $-+$}
  \label{fig:id-caps-cups}
\end{figure}

Our category also carries a strict pivotal structure (Definition~\ref{def:pivotal}).
The dual $w^*$ of an object $w$ is obtained by reversing it and flipping the polarities. For instance,
$(+--)^*$ is $++-$. Going by this definition, note we also have $(w^*)^* = w$.
We also have natural transformations
$\eta_A : \unit \to A \tensor A^*$ and $\varepsilon_A : A^* \tensor A \to \unit$
that we picture in Figure~\ref{fig:id-caps-cups}.
They satisfy the yanking equations, which gives us in particular
the closed structure by setting $A \lin B = B \tensor A^*$,
$\ev_{A,B} = \id_B \tensor \varepsilon_A$ and $\Lambda_{A,B,C}(f) = (\id_A \tensor \eta_B) ; (f \tensor \id_{B^*})$ as per Lemma~\ref{lem:pivotal-closed}.

Finally, observe that we may define a natural order on combinatorial diagrams
sharing the same vertices. Given two such diagrams $d$ and $d'$ with respective
edge sets $E_d$ and $E_d'$, we say that $d \le d'$ whenever $E_d \subseteq E_d'$
and their edge labellings coincide over $E_d$.
This gives an order on hom-sets of $\catTransition_{\Sigma}$ where composition
and tensoring are easily checked to be both monotone. Together with the
observation that we have cups and caps that satisfy the yanking equations, we
thus have.

\begin{lemma}
$\catTransition_\Sigma$ equipped with the concatenating tensor and inclusion
of labelled edges is a strict monoidal-closed poset-enriched category.
\end{lemma}

Finally, we note that, for any set of vertices, the bottom element in this
order we have defined over diagrams is given by the graph with no edges.
Tensoring bottom elements yield bottom elements and $\id_\unit$ is the bottom
element of $\Hom{\catTransition_\Sigma}{\unit}{\unit}$.

\subsection{Two-way planar transducers}

Following Colcombet and Petrişan~\cite{colcombetpetrisan}, we formally define our
notion of \emph{two-way planar reversible transducers} (2PRFTs) as being functors whose
domain $\catAutomata{\Sigma}$ is category whose morphisms represent infixes of
words. In our situation it will mostly have the advantage of concision and
making the relationship between 2PRFTs and $\catTransition$ obvious.

\begin{definition}
  \label{def:catAutomata}
    For any finite alphabet $\Sigma$, there is a three object category
    \emph{$\catAutomata{\Sigma}$} generated by the following finite graph,
    where there is one morphism for each $a\in \Sigma$.
    \begin{center}
      \begin{tikzcd}
        \sIn \arrow[r, "\stringStart"] & \sState \arrow[loop above, "a"] \arrow[r, "\stringEnd"] & \sOut
      \end{tikzcd}
    \end{center}
    Morphisms $\sState \to \sState$ are identified with words of $\Sigma^*$
    by writing $au$ for $a ; u$ and $\emptyString$ for $\id_\sState$ (note that
    the composition is left-to-right).
    For any category $\cC$ and objects $I$ and $O$ of $\cC$, define a
  \emph{$(\cC, I, O)$-automaton with input alphabet $\Sigma$} to be a functor $\cA : \catAutomata{\Sigma} \to \cC$
  with $\cA(\sIn) = I$ and $\cA(\sOut) = O$.
  Given such an automaton $\cA$, its semantics is the map $\Sigma^* \to \Hom{\cC}{I}{O}$
  given by $w \mapsto \cA(\stringStart) ; \cA(w) ; \cA(\stringEnd)$.
\end{definition}

In this framework, we can for instance define deterministic finite automata
as $(\FinSet, 1, 2)$-functors and nondeterministic ones as $(\FinRel, 1, 1)$-functors,
and check that the semantics computes the languages as we expect it.
As the category $\catTransition_{\Sigma}$ corresponds to transition profiles
as studied in \cite{planartrans}, we will use that to define 2PRFTs in a completely
equivalent way. In that case, we will pick $I$ and $O$ such that
$\Hom{\catTransition_\Sigma}{I}{O} \cong \Sigma^*_\bot$, where $\Sigma^*_\bot$
is the disjoint union of $\Sigma^*$ with a singleton containing a $\bot$ element;
this is required because we will obtain this function by reading off the label of
a specific edge of a morphism that may not always exist.

\begin{example}
  \label{ex:2RFT}
Let us build a $(\catTransition_{\{0,1,2\}}, +, +)$-automaton with input alphabet $\{0,1,2\}$ that
pads any string in $\{0,1,2\}^*$ to ensure that every 2 is preceded by a 1 in
the output. First, here is a standard automata-theoretic picture of such a device and its transition table:\\
 {\small   \begin{multicols}{2}
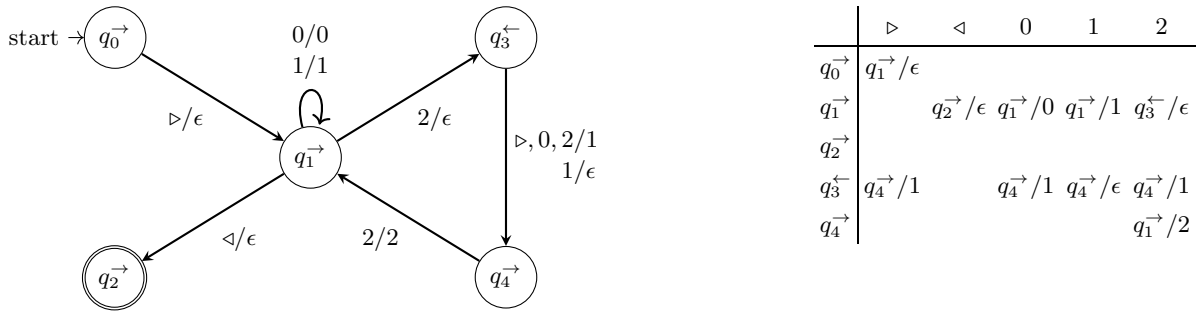
 \begin{center}
   \begin{tikzpicture}[scale=0.5, auto]
      \node (start) [state, initial] {$\forward{q_0}$};
      \node (main) [state, below right= 1 and 2 of start] {$\forward{q_1}$};
      \node (end) [state, below left= 1 and 2 of main, accepting] {$\forward{q_2}$};
      \node (backup) [state, above right= 1 and 2 of main] {$\backward{q_3}$};
      \node (eat2) [state, below right= 1 and 2 of main] {$\forward{q_4}$};
      \path [-stealth, thick]
      (start) edge node [swap] {$\stringStart \ioSeparator \emptyString$} (main)
      (main) edge node  {$\stringEnd \ioSeparator \emptyString$} (end)
      edge [loop above] node {$
        \begin{aligned}
          0 &\ioSeparator 0 \\
          1 &\ioSeparator 1
        \end{aligned}
        $} ()
      edge node [swap] {$2 \ioSeparator \epsilon$} (backup)
      (backup) edge node {$
        \begin{aligned}
          \stringStart,0,2 &\ioSeparator 1 \\
          1 &\ioSeparator \epsilon
        \end{aligned}
        $} (eat2)
      (eat2) edge node  {$2 \ioSeparator 2$} (main);
    \end{tikzpicture}
\end{center}
\columnbreak
  \[ 
  \begin{array}{c|ccccc}
    & \stringStart & \stringEnd & 0 & 1 & 2 \\
    \hline
    \forward{q_0} & \forward{q_1} \ioSeparator \emptyString& & & & \\
    \forward{q_1} & & \forward{q_2} \ioSeparator \emptyString & \forward{q_1} \ioSeparator 0 & \forward{q_1} \ioSeparator 1 & \backward{q_3} \ioSeparator \emptyString\\
    \forward{q_2} & & & & & \\
    \backward{q_3} & \forward{q_4} \ioSeparator 1 & & \forward{q_4} \ioSeparator 1 & \forward{q_4} \ioSeparator \emptyString & \forward{q_4} \ioSeparator 1\\
    \forward{q_4} & & & & & \forward{q_1} \ioSeparator 2 \\
  \end{array}
\]
\end{multicols}}

  Using the ordering given by the subscripts and assigning $+$ to the
  forward vertices, i.e., $\forward{q_i}$, and $-$ to the backward
  vertices, i.e., $\backward{q_j}$, we obtain the word
  $F(\mathrm{states}) = +++-+$. For each letter $a \in \Sigma\sqcup\{\stringStart,\stringEnd\}$
  we assign the morphism $F(a)$ of the functor by reading it off the table.

  
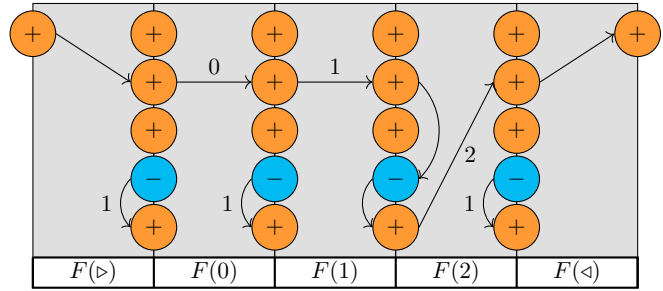
\begin{center}
    \begin{tikzpicture}[->,scale=0.8, every node/.style={transform shape}]
      \draw[thick,-] (2,0.5) -- (12,0.5) -- (12,0) -- (2,0) -- (2,0.5);
      \foreach \col in {1,...,6} {
        \draw[thick,-] (2*\col,0) -- (2*\col,0.5);
      }
      \node at (3,0.25) {$F(\stringStart)$};
      \node at (5,0.25) {$F(0)$};
      \node at (7,0.25) {$F(1)$};
      \node at (9,0.25) {$F(2)$};
      \node at (11,0.25) {$F(\stringEnd)$};

      \node (in) at (2,4.2) [plus];
      \foreach \col in {2,...,5} {
        \node (q0\col) at (2*\col,4.2) [plus];
        \node (q1\col) at (2*\col,3.4) [plus];
        \node (q2\col) at (2*\col,2.6) [plus];
        \node (q3\col) at (2*\col,1.8) [minus];
        \node (q4\col)  at (2*\col,1) [plus];
      }
      \node (out) at (2*6,4.2) [plus];

      \draw[->] (in.east) to (q12.west) ;
      \draw[->] (q32.west) to [bend right = 45] node [left] {$1$} (q42.west);
      \draw[->] (q12.east) to node [above] {$0$} (q13.west) ;
      \draw[->] (q33.west) to [bend right = 45] node [left] {$1$} (q43.west);
      \draw[->] (q13.east) to node [above] {$1$} (q14.west) ;
      \draw[->] (q34.west) to [bend right = 45] (q44.west);
      \draw[->] (q14.east) to [bend left = 45] (q34.east) ;
      \draw[->] (q35.west) to [bend right = 45] node [left] {$1$} (q45.west);
      \draw[->] (q44.east) to node [right] {$2$} (q15.west);
      \draw[->] (q15.east) to (out.west) ;

      \begin{scope}[on background layer]
        \draw[box] ([yshift=.5cm]in) rectangle ([yshift=-.5cm]q42);
        \draw[box] ([yshift=.5cm]q02) rectangle ([yshift=-.5cm]q43);
        \draw[box] ([yshift=.5cm]q03) rectangle ([yshift=-.5cm]q44);
        \draw[box] ([yshift=.5cm]q04) rectangle ([yshift=-.5cm]q45);
        \draw[box] ([yshift=.5cm]out) rectangle ([yshift=-.5cm]q45);
      \end{scope}
    \end{tikzpicture}
  \end{center}
\end{example}

\begin{definition}
  \label{def:2PRFT}
  A \emph{two-way planar reversible transducer (2PRFT)} $\cT$ with input
  alphabet $\Sigma$ and output alphabet $\Gamma$ is a
  $(\catTransition_\Gamma, \emptyString, +-)$-automaton with input alphabet $\Sigma$.

  Writing $\Gamma^*_\bot$ for the disjoint union of $\Gamma^*$ with a singleton
  $\{\bot\}$ containing a designated $\bot$ element, the semantics of
  such a 2PRFT $\cT$ induces a function
    \begin{center}
      \begin{tikzcd}[column sep=5cm]
        \Sigma^* \arrow{r}{\text{semantics of $\cT$}} & \Hom{\catTransition_\Gamma}{\emptyString}{+-}
        \arrow{r}{\text{read off the label}}[swap]{\text{($\bot$ if there is no edge)}} & \Gamma^*_\bot
      \end{tikzcd}
    \end{center}
\end{definition}

Note that our choice of $\emptyString$ and $+-$
means that by convention, both ``initial'' and ``final'' states must occur
before the initial and after the final reading of $\triangleleft$, while the
convention of~\cite[Definition 2.1]{planartrans} and in Example~\ref{ex:2RFT} is slightly different for the
initial state. In that version, it should start by reading $\triangleright$,
making the 2PRFTs of~\cite{planartrans} isomorphic to
$(\catTransition_{\Gamma}, +, +)$-automata
rather than $(\catTransition_{\Gamma}, \emptyString, +-)$-automata.
But it is not hard to see that both options induce the same class of
string-to-string functions. It will turn out that Definition~\ref{def:2PRFT} matches
much more closely $\lambda$-transducers, so we favor it out of convenience.

\section{Equivalence between planar transducers and $\lampa$ for strings}
\label{sec:main}

Now that we have introduced properly our two classes of string-to-string
functions, affine $\lampa$-definable functions and first-order transductions,
as well as two formalisms that define them, $\lampa$-transducers and 2PRFTs,
we will now embark on the proof that they are equivalent.

\mainplanarstring*

To prove that affine $\lampa$-definable functions are first-order transduction, we
use the fact that the former class correspond to $\lampa$-transductions and
then define a map from $\lampa$-transductions to 2PRFTs that preserves the
semantics. To do so, we define an interpretation of purely affine
$\lampa$-terms (with duplicable free variables in $\underline{\Gamma}$)
in the category $\catTransition_\Gamma$.
One difficulty is that $\catTransition_\Gamma$ is not \emph{affine} monoidal
closed, that is, $\unit$ is not a terminal object. So instead of terminal
maps we will use $\bot_A \in \Hom{\catTransition_\Gamma}{A}{\unit}$ and
establish that $\beta$-reductions correspond to inequalities in $\catTransition_\Gamma$
in Subsection~\ref{subsec:interplambda}.
We will then conclude in Subsection~\ref{subsec:lambda22PRFT}.
Proving the converse, which will amount to a coding exercise and a reference to
\cite{aperiodic} once the right characterization of first-order transductions
as compositions of more basic functions is recalled, will be done in
Subsection~\ref{subsec:fotrans2lambda}.

\subsection{Interpreting $\lampa$}
\label{subsec:interplambda}

All results of this subsection hold for any strict
monoidal-closed poset-enriched category $\cC$ with a family of least
elements $\bot_X \in \Hom{\cC}{X}{\unit}$ stable under $\tensor$ and with
$\bot_\unit = \id_\unit$, provided we are given an
object $\interp{\basety}$ of $\cC$ and, for every constant $x : \tau$ in
$\underline{\Gamma}$ a suitable interpretation $\interp{x} : \unit \to \interp{\tau}$,
where $\interp{\tau}$ is extended inductively over all types by setting for
$\interp{\tau \lolly \sigma}$ a chosen internal hom $\interp{\tau} \lolly \interp{\sigma}$.
This interpretation also extends to contexts by tensoring as usual by setting
$\interp{\cdot} = \unit$ and $\interp{\Delta, x : \tau} = \interp{\Delta} \tensor \interp{\tau}$.
The extension of $\interp{-}$ over all purely affine $\lampa$ typing derivations\footnote{It can
actually be shown that the interpretation of a legal typing derivation
$\underline{\Gamma}; \Delta \vdash t : \tau$ only depends on the conclusion.
But we won't need to make use of that fact.}
is then given in Figure~\ref{fig:interp-mccbot}.
One thing to note is that the overall interpretation $\interp{t}$ of a term $t$ can be carried out in
polynomial time in the size of $t$ because type-checking is polynomial-time and
composition in $\catTransition_\Gamma$ can be performed in logarithmic space.

\begin{figure}
\[
\begin{array}{ccc}
  \inferrule*{\text{$x$ a variable of $\underline{\Gamma}$}}
            {\judgement{\underline{\Gamma}}{\Delta}{x}{\tau}}
  & \longmapsto &  \interp{x} \circ \bot_{\interp{\Delta}}
  : \interp{\Delta} \to \interp{\tau}
  \\ \\
\inferrule*{ }
            {\judgement{\underline{\Gamma}}{\Delta, \hasType{x}{\tau}, \Delta'}{x}{\tau}}
  & \longmapsto &
  \bot_{\interp{\Delta}}
  \tensor \id_{\interp{\tau}}
  \tensor
  \bot_{\interp{\Delta'}} : \interp{\Delta} \tensor \interp{\tau} \tensor \interp{\Delta'} \to \interp{\tau}
  \\
  \\
\inferrule*{\judgement{\underline{\Gamma}}{\Delta, \hasType{x}{\tau}}{t}{\sigma}}
           {\judgement{\underline{\Gamma}}{\Delta}{\lambda{}x.t}{\tau\lolly\sigma}}
  & \longmapsto &
\inferrule*{\interp{t} : \interp{\Delta} \tensor \interp{\tau} \to \interp{\sigma}}
  {
  \Lambda_{\interp{\Delta},\interp{\tau},\interp{\sigma}
  }(\interp{t})
  : \interp{\Delta} \to \interp{\tau} \lolly \interp{\sigma}}
\\\\
\inferrule*{\judgement{\underline{\Gamma}}{\Delta}{t}{\tau\lolly\sigma} \and
            \judgement{\underline{\Gamma}}{\Delta^{\prime}}{u}{\tau}}
            {\judgement{\underline{\Gamma}}{\Delta, \Delta^{\prime}}{t \; u}{\sigma}}
  & \longmapsto &
\inferrule*{\interp{t} : \interp{\Delta} \to \interp{\tau} \lolly \interp{\sigma}
   \and
   \interp{u} : \interp{\Delta'} \to \interp{\tau}
  }
  {
    \ev_{\interp{\tau},\interp{\sigma}} \circ (\interp{t} \tensor \interp{u})
  : \interp{\Delta} \tensor \interp{\Delta'} \to \interp{\sigma}}
\end{array}
\]
\caption{Interpretation of purely affine $\lampa$-terms over $\underline{\Gamma}$ (parameterized by $\interp{\basety}$ and $\interp{x} : \unit
  \to \interp{\tau}$ for $x
  : \tau$ occurring in $\underline{\Gamma}$).}
\label{fig:interp-mccbot}
\end{figure}

While we will not have that $t =_{\beta\eta} u$ implies
$\interp{t} = \interp{u}$, it will be the case that:
\begin{itemize}
  \item $\eta$-equivalences $t =_\eta u$ will be mapped to equalities of morphisms $\interp{t} = \interp{u}$
  \item $\beta$-reductions $t \to_\beta u$ will be mapped to inequalities $\interp{t} \ge \interp{u}$
\end{itemize}
so that, in particular, a normal form $t_\NF$ of $t$ will always satisfy
$\interp{t_\NF} \le \interp{t}$.
Let us now establish that, beginning with $\eta$-equivalence.

\begin{lemma}
\label{lem:interp-eta-eq}
When $\underline{\Gamma}; \Delta \vdash t : \tau \lolly \sigma$, we have $\interp{\lambda x. \; t \; x} = \interp{t}$.
\end{lemma}
\begin{proof}
By definition $\interp{\lambda x. \; f \; x} = \Lambda_{\interp{\Delta},\interp{\tau}, \interp{\sigma}}(\ev_{\interp{\tau},\interp{\sigma} \circ
(\interp{f} \tensor \id_{\interp{\tau}})})$, and the latter is equal to
$\interp{t}$ by using the universal property of the internal hom.
\end{proof}

\begin{corollary}
\label{cor:interp-eta-eq}
If we have $t =_\eta u$, then we have that $\interp{t} = \interp{u}$.
\end{corollary}
\begin{proof*}{Proof idea}
Easy induction using Lemma~\ref{lem:interp-eta-eq}.
\end{proof*}

\begin{restatable}{lem}{interpsubstle}
\label{lem:interp-subst-le}
Suppose we have $\judgement{\underline{\Gamma}}{\Delta, \hasType{x}{\tau}, \Delta'}{t}{\sigma}$
and
$\judgement{\underline{\Gamma}}{\Delta''}{u}{\tau}$. Then we have
  \[ \interp{t[u/x]} \le \interp{t} \circ (\id_{\interp{\Delta}} \tensor \interp{u} \tensor \id_{\interp{\Delta'}}) \qquad \qquad {\footnotesize (~~: \interp{\Delta, \Delta'',
  \Delta'} \to \interp{\tau})}\]
\end{restatable}
\begin{proof}
The proof is by induction over the typing derivation of $t$. We will use throughout
that $\circ$ and $\tensor$ are monotone, that $\bot_A \tensor \bot_B = \bot_{A \tensor B}$ as well as $\id_\unit = \bot_\unit$ and that $\bot_A \le f$ for any $f : \unit \to A$
without calling explictly attention to it.
\begin{itemize}
\item If $t$ is the variable $x$, then both sides are equal to
  $\bot_{\interp\Delta} \tensor \interp{u} \tensor \bot_{\interp{\Delta'}}$.
\item If $t$ a variable other than $x$ from the linear part of the context,
  say $y$ from $\Delta$ such that we have $\Delta = \Theta, y : \sigma, \Theta'$
    (the case where $y$ is from $\Delta'$ is treated analogously), we derive the following using that $\bot_{\interp{\Delta''}} \le \bot_{\interp{\sigma}}
    \circ \interp{u}$:
\[
\begin{array}{lcl}
  \interp{y[u/x]} &=& \bot_{\interp{\Theta}} \tensor \id_{\interp{\sigma}} \tensor \bot_{\interp{\Theta', \Delta'', \Delta'}}\\
  &=&
  \bot_{\interp{\Theta}} \tensor \id_{\interp{\sigma}} \tensor \bot_{\interp{\Theta'}} \tensor \bot_{\Delta''} \tensor \bot_{\interp{\Delta'}}\\
&\le&
  \bot_{\interp{\Theta}} \tensor \id_{\interp{\sigma}} \tensor \bot_{\interp{\Theta'}} \tensor (\bot_{\interp{\sigma}} \circ \interp{u}) \tensor \bot_{\interp{\Delta'}}\\
&=&
  \bot_{\interp{\Theta}} \tensor \id_{\interp{\sigma}} \tensor \bot_{\interp{\Theta'}} \tensor (\bot_{\interp{\sigma}} \circ \interp{u}) \tensor \bot_{\interp{\Delta'}}\\
&=&
  (\bot_{\interp{\Theta}} \tensor \id_{\interp{\sigma}} \tensor \bot_{\interp{\Theta'}} \tensor \bot_{\interp{\sigma}} \tensor \bot_{\interp{\Delta'}})
  \circ (\id_{\interp{\Delta}} \tensor \interp{u} \tensor \id_{\interp{\Delta'}})\\
&=&
  \interp{y}
  \circ (\id_{\interp{\Delta}} \tensor \interp{u} \tensor \id_{\interp{\Delta'}})\\
\end{array}
\]
\item If $t$ is a variable of $\underline{\Gamma}$, the desired inequality follows
  from
\[\bot_{\interp{\Delta, \Delta'', \Delta'}} \le \bot_{\interp{\Delta} \tensor \interp{\tau} \tensor \interp{\Delta'}} \circ (\id_{\interp{\Delta}} \tensor
    \interp{u} \tensor \id_{\interp{\Delta'}})\]
\item If $t = f \; g$ for $f : \tau \lolly \sigma$, $g : \tau$, we have two subcases
  according to which context $x$ appears in.
  \begin{itemize}
    \item Suppose $x$ appears in the context of $f$ so that we have, $\Delta' = \Delta'_f, \Delta'_g$ and
  judgements
  \[\judgement{\underline{\Gamma}}{\Delta, \hasType{x}{\tau}, \Delta_f'}{f}{\tau \lolly \sigma} \qquad \text{and} \qquad
      \judgement{\underline{\Gamma}}{\Delta'_g}{g}{\tau}\]
By the induction hypothesis, we have
$\interp{f[u/x]} \le
\interp{f} \circ (\id_{\interp{\Delta}} \tensor \interp{u} \tensor \id_{\interp{\Delta_f'}})
$, which allows to derive
\[
\begin{array}{lcl}
  \interp{t[u/x]} &=& \interp{f[u/x] \; g} \\
  &=&
  \ev_{\interp{\tau}, \interp{\sigma}} \circ \left(\interp{f[u/x]} \tensor \interp{g}\right)
  \\
  &\le&
\ev_{\interp{\tau}, \interp{\sigma}} \circ \left(
\left(\interp{f} \circ (\id_{\interp{\Delta}} \tensor \interp{u} \tensor \id_{\interp{\Delta_f'}})\right)
  \tensor \interp{g}\right)\\
&=&
\ev_{\interp{\tau}, \interp{\sigma}} \circ \left(
  \left(\interp{f} \circ (\id_{\interp{\Delta}} \tensor \id_{\interp{\tau}} \tensor \id_{\interp{\Delta_f'}})\right)
  \tensor \interp{g}\right) \circ
(\id_{\interp{\Delta}} \tensor \interp{u} \tensor \id_{\interp{\Delta_f'}} \tensor \id_{\interp{\Delta_g'}})\\
&=&
\ev_{\interp{\tau}, \interp{\sigma}} \circ \left(
\interp{f} \circ (\id_{\interp{\Delta, x : \tau, \Delta_f'}}
  \tensor \interp{g})\right) \circ
(\id_{\interp{\Delta}} \tensor \interp{u} \tensor \id_{\interp{\Delta_f'}} \tensor \id_{\interp{\Delta_g'}})\\
&=&
 \interp{f \; g} \circ (\id_{\interp{\Delta}} \tensor \interp{u} \tensor \id_{\interp{\Delta_f'}} \tensor \id_{\interp{\Delta_g'}})\\
\end{array}
\]
  \item The
  case when $x$ appears in the context of $g$ is very similar and left to the reader.
  \end{itemize}
\item If $t[u/x] = \lambda y. \; t'[u/x]$ with $y \neq x$, then
the premise of the rule under consideration is $\judgement{\underline{\Gamma}}{\Delta,
    \hasType{x}{\tau}, \Delta', \hasType{y}{\tau'}}{t}{\sigma}$ and the induction
    hypothesis thus is
\[\interp{t'[u/x]} \le
\interp{t'} \circ (\id_{\interp{\Delta}} \tensor
  \interp{u} \tensor \id_{\interp{\Delta'}} \tensor \id_{\interp{\tau'}})
\]
So the result is then derived as follows, using the monotonicity of $\Lambda$ and
that we have $\Lambda_{A,B,C}(h \circ (\ell \tensor \id_B)) = \Lambda_{A,B,C}(h) \circ \ell$
in monoidal closed categories:
\[
\begin{array}{lcl}
  \interp{t[u/x]} &=& \interp{\lambda y. \; t'[u/x]} \\
  &=& \Lambda_{\interp{\Delta} \tensor \interp{\tau} \tensor \interp{\Delta'}, \interp{\tau'}, \interp{\sigma}}(\interp{t'[u/x]})\\
  &\le&
\Lambda_{\interp{\Delta} \tensor \interp{\tau} \tensor \interp{\Delta'}, \interp{\tau'}, \interp{\sigma}}
(
\interp{t'} \circ (\id_{\interp{\Delta}} \tensor
  \interp{u} \tensor \id_{\interp{\Delta'}} \tensor \id_{\interp{\tau'}})
)\\
&=&
\Lambda_{\interp{\Delta} \tensor \interp{\tau} \tensor \interp{\Delta'}, \interp{\tau'}, \interp{\sigma}}
(\interp{t'})
\circ (\id_{\interp{\Delta}} \tensor \interp{u} \tensor \id_{\interp{\Delta'}})
)
  \\
&=&
  \interp{\lambda y. \; t'} 
\circ (\id_{\interp{\Delta}} \tensor \interp{u} \tensor \id_{\interp{\Delta'}})
)
  \\
\end{array}
\]
\end{itemize}
\end{proof}

\begin{corollary}
\label{cor:interp-beta-le}
If we have $t \to_\beta u$, then we have that $\interp{t} \ge \interp{u}$.
\end{corollary}
\begin{proof*}{Proof idea}
Easy induction using monotonicity of $\circ$ and $\tensor$ together with Lemma~\ref{lem:interp-subst-le}.
\end{proof*}

We can thus conclude with the only information we will need in the next
subsection.

\begin{corollary}
\label{cor:interp-normal}
For any $t$ whose normal form is $t_\NF$, we have $\interp{t_\NF} \le \interp{t}$.
\end{corollary}

\subsection{From $\lampa$-transducers to 2PRFTs}
\label{subsec:lambda22PRFT}

Now we fix an output alphabet $\Gamma$ for the $\lampa$-transducer.
We shall then use the interpretation from the previous subsection with
$\cC = \catTransition_\Gamma$, $\interp{\basety} = +-$ and the interpretation
of the constants of $\underline{\Gamma}$ given in Figure~\ref{fig:generators}.

\begin{restatable}{lem}{interpwords}
\label{lem:interp-words}
For $w \in \Gamma^*$,  $\underline{\Gamma}; \cdot \vdash \underline{w} : \basety$ is
interpreted by the diagram
  
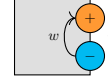
\begin{tikzpicture}[on grid,->, scale = 0.5, transform shape, baseline=(current bounding box.center)]
      \node (in first) at (0,0) [draw opacity=0] {};
      \node (in last) [below=of in first] [draw opacity=0] {};

      \node (out 1) [right=2 of in first][plus];
      \node (out 2) [below=of out 1] [minus];
      \draw (out 2) to [bend left=90] node [left] {$w$} (out 1);

      \begin{scope}[on background layer]
        \draw[box] ([yshift=.5cm]in first) rectangle ([yshift=-.5cm]out 2);
      \end{scope}
    \end{tikzpicture}.
\end{restatable}
\begin{proof}
This is done by an induction over $w$. When $w = \emptyString$, this is obvious.
When $w = aw'$, we have 
  $\interp{\underline{aw'}} = \interp{a \; \underline{w'}} = (\interp{a} \tensor \interp{\underline{w'}}) ; \ev_{\interp{\basety},\interp{\basety}}$.
Applying the induction hypothesis and drawing out the picture of this composition,
we can conclude by chasing the path.
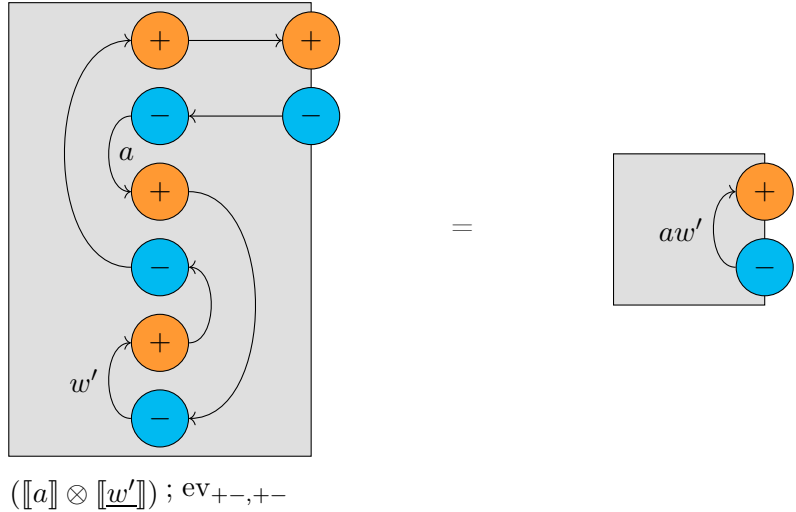
\begin{center}
    \begin{tikzpicture}
      [on grid,->]
      \node (in first) at (0,0) [draw opacity=0] {};
      \node (in last) [below=5 of in first] [draw opacity=0] {};
      \node (in last last) [below=of in last] [draw opacity=0] {};
      \node (out last) [right=4 of in last] [draw opacity=0] {};

      \node (out 1) [right=2 of in first][plus];
      \node (out 2) [below=of out 1] [minus];
      \node (out 3) [below=of out 2] [plus];
      \node (out 4) [below=of out 3] [minus];

      \draw (out 4) to [bend right=270] (out 1);
      \draw (out 2) to [bend right=90] node [right] {$a$} (out 3);

      \node (out 5) [below=of out 4][plus];
      \node (out 6) [below=of out 5] [minus];

      \draw (out 6) to [bend right=270] node [left] {$w'$} (out 5);

      \node (outout 1) [right=2 of out 1] [plus];
      \node (outout 2) [below=of outout 1] [minus];
      \node (outout 3) [below=1.5 of outout 2] [draw opacity=0] {};
      \node (eq) [right=2 of outout 3] {=};

      \node (outout 25) [below=of outout 2] [draw opacity=0] {};

      \node (in2 first) [right=4 of outout 25] [draw opacity=0] {};
      \node (in2 last) [below=of in2 first] [draw opacity=0] {};

      \node (out2 1) [right=2 of in2 first][plus];
      \node (out2 2) [below=of out2 1] [minus];
      \draw (out2 2) to [bend left=90] node [left] {$aw'$} (out2 1);

      \begin{scope}[on background layer]
        \draw[box] ([yshift=.5cm]in2 first) rectangle ([yshift=-.5cm]out2 2);
      \end{scope}

      \draw (out 5) to [bend right=90] (out 4);
      \draw (out 3) to [bend left=90] (out 6);
      \draw (out 1) to (outout 1);
      \draw (outout 2) to (out 2);
      \node (tensor) [right=of in last last] {$(\interp{a} \tensor \interp{\underline{w'}})$};
      \node (punc) [right=of tensor] {~~;};
      \node (ev) [right=of punc] {$\ev_{+-,+-}$};
      \begin{scope}[on background layer]
        \draw[box] ([yshift=.5cm]in first) rectangle ([yshift=-.5cm]out last);
      \end{scope}
    \end{tikzpicture}
\end{center}
\end{proof}

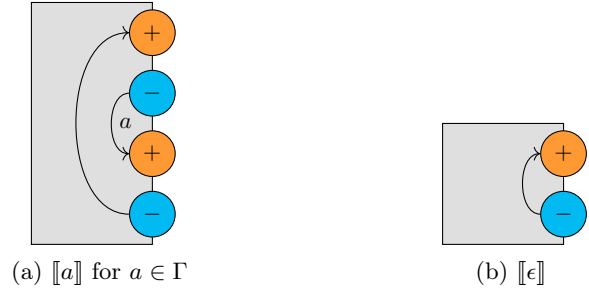
\begin{figure}
  \centering
  \begin{subfigure}{0.3\linewidth}
    \centering
    \begin{tikzpicture}
      [on grid,->, scale = 0.8, transform shape, baseline=(current bounding box.center)]
      \node (in first) at (0,0) [draw opacity=0] {};
      \node (in last) [below=3 of in first] [draw opacity=0] {};

      \node (out 1) [right=2 of in first][plus];
      \node (out 2) [below=of out 1] [minus];
      \node (out 3) [below=of out 2] [plus];
      \node (out 4) [below=of out 3] [minus];

      \draw (out 4) to [bend right=270] (out 1);
      \draw (out 2) to [bend right=90] node [right] {$a$} (out 3);

      \begin{scope}[on background layer]
        \draw[box] ([yshift=.5cm]in first) rectangle ([yshift=-.5cm]out 4);
      \end{scope}
    \end{tikzpicture}
    \caption{$\interp{a}$ for $a \in \Gamma$}
  \end{subfigure}
  \begin{subfigure}{0.3\linewidth}
    \centering
    \begin{tikzpicture}
      [on grid,->, scale = 0.8, transform shape, baseline=(current bounding box.center)]
      \node (in first) at (0,0) [draw opacity=0] {};
      \node (in last) [below=1 of in first] [draw opacity=0] {};

      \node (out 1) [right=2 of in first][plus];
      \node (out 2) [below=of out 1] [minus];

      \draw (out 2) to [bend right=270] (out 1);

      \begin{scope}[on background layer]
        \draw[box] ([yshift=.5cm]in first) rectangle ([yshift=-.5cm]out 2);
      \end{scope}
    \end{tikzpicture}
    \caption{$\interp{\epsilon}$}
  \end{subfigure}
  \caption{Interpretation of constants as diagrams}
  \label{fig:generators}
\end{figure}

\begin{theorem}
Every $\lampa$-transducer can be converted into an equivalent 2PRFT in polynomial time.
\end{theorem}
\begin{proof}
As per Definition~\ref{def:normalAffineStr2Str}, assume that we have a purely
affine iteration type $\kappa$, terms
$\judgement{\underline{\Gamma}}{\nullEnv}{d_a}{\kappa \lolly \kappa}$
for each $a \in \Sigma$,
$\judgement{\underline{\Gamma}}{\nullEnv}{o}{\kappa \lolly \basety}$ and
$\judgement{\underline{\Gamma}}{\nullEnv}{d_\epsilon}{\basety \lolly
  \kappa}$. Using the semantic interpretation given above, we obtain the respective
  morphisms $\interp{d_a} : \unit \to \interp{\kappa}\lolly\interp{\kappa}$ (for each $a \in \Sigma$),
$\interp{o} : \unit \to \interp{\kappa} \lolly +-$ and
$\interp{d_\epsilon} : \unit \to \interp{\kappa}$.
We define the equivalent 2PRFT $\cT$ on the generating morphisms of $\catAutomata{\Sigma}$ as follows.
\[\cT(a) ~~=~~
  \Lambda^{-1}_{\unit, \interp{\kappa},\interp{\kappa}}(\interp{d_a})
\qquad
\cT(\triangleleft) ~~=~~
\Lambda^{-1}_{\unit,\interp{\kappa},\interp{\basety}}(\interp{o})
\qquad \text{and} \qquad
\cT(\triangleright) =
\interp{d_\epsilon}
\]
To prove that $\cT$ computes the same function as the $\lampa$-transducer given,
let's consider the diagram below.
\[\begin{tikzcd}
	{\Sigma^*} &&& {\{t \mid \underline{\Gamma};\cdot \vdash t : \basety\}} && {\{\underline{w} \mid w \in\Gamma^*\}} & {\Gamma^*} \\
	\\
	&&&&& {\Hom{\catTransition_\Gamma}{\unit}{+-}} & {\Gamma^*_\bot}
	\arrow["{\text{normalize}}", from=1-4, to=1-6]
	\arrow["{\interp{-}}"'{pos=0.3}, curve={height=18pt}, from=1-4, to=3-6]
	\arrow["{\interp{-}}", from=1-6, to=3-6]
	\arrow["\cT"', curve={height=24pt}, from=1-1, to=3-6]
	\arrow["{w \mapsto o\; (d_{w_n} \; (\ldots d_\epsilon) \ldots))}", from=1-1, to=1-4]
	\arrow["{\Large{\ge}}"{description}, curve={height=-6pt}, draw=none, from=1-4, to=3-6]
	\arrow["{(1)}"{description}, draw=none, from=1-1, to=3-6]
	\arrow["{(2)}"{description}, draw=none, from=1-6, to=3-7]
	\arrow["\cong"{description}, from=3-6, to=3-7]
	\arrow["\cong"{description}, from=1-6, to=1-7]
	\arrow["{{\subseteq}}"{description}, from=1-7, to=3-7]
\end{tikzcd}\]

By inspecting the definitions, the map defined by the
$\lampa$-transducer is obtained by following the topmost maximal path while the map
defined by $\cT$ is given by the bottommost maximal path, which we must argue
define the same map. To do so, it suffices to show that faces (1)
and (2) commute while the central face denotes an inequality between maps; here
all nodes are equipped with an order structure by taking the discrete order
for the objects on the top row, the order from the enriched structure
of $\catTransition_\Gamma$ for $\Hom{\catTransition_\Gamma}{\unit}{+-}$ and
by taking for $\Gamma_\bot^*$ the minimal order such that $\bot \le w$ for $w \in \Gamma^*$.
Then the maps are ordered by pointwise ordering.
That the inequality ``top path $\le$ bottom path'' suffices to derive
``top path = bottom path'' is because the top path necessarily is a maximal
element for the pointwise ordering of maps. This is due to the fact that
$\Gamma^*$ consists of the maximal elements of $\Gamma^*_\bot$.

That (2) commutes is exactly the statement of Lemma~\ref{lem:interp-words} while
the inequality in the central face is Corollary~\ref{cor:interp-normal}.
All that remains to be proven is that (1) commutes. This is witnessed by the
chain of equations below for a fixed input word
$w = w_1 \ldots w_n \in \Sigma^*$.
\[
  \begin{array}{lclr}
\cT(\triangleright w \triangleleft)
&=&
\cT(\triangleleft) \circ \cT(w_n) \circ \ldots \circ \cT(w_1) \circ \cT(\triangleright)
& {\footnotesize \text{(by functoriality)}}
\\
&=&
\Lambda^{-1}_{\unit,\interp{\kappa},\interp{\basety}}(\interp{o}) \circ
    \Lambda^{-1}_{\unit, \interp{\kappa},\interp{\kappa}}(\interp{d_{w_n}})
\circ \ldots \circ
    \Lambda^{-1}_{\unit, \interp{\kappa},\interp{\kappa}}(\interp{d_{w_1}})
\circ
\interp{d_\epsilon}
& {\footnotesize \text{(by definition of $\cT$)}}
\\
&=&
    \multicolumn{2}{l}{
\ev_{\interp{\kappa},\interp{\basety}} \circ (\interp{o} \tensor \id_{\interp{\kappa}}) \circ
\ev_{\interp{\kappa}, \interp{\kappa}} \circ (\interp{d_{w_n}} \tensor \id_{\interp{\kappa}})
\circ \ldots \circ
\ev_{\interp{\kappa}, \interp{\kappa}} \circ (\interp{d_{w_1}} \tensor \id_{\interp{\kappa}})
\circ
    \interp{d_\epsilon}}
    \\&&\multicolumn{2}{r}{\footnotesize\text{(because $\Lambda^{-1}_{A,B,C}(f) = \ev_{B,C} \circ (f \tensor \id_B)$})}
\\
&=&
    \multicolumn{2}{l}{
\ev_{\interp{\kappa},\interp{\basety}} \circ (\interp{o} \tensor 
(\ev_{\interp{\kappa}, \interp{\kappa}} \circ (\interp{d_{w_n}} \tensor \ldots
    \ev_{\interp{\kappa}, \interp{\kappa}} \circ (\interp{d_{w_1}} \tensor \interp{d_\epsilon})\ldots)))
}
\\&&&
    {\footnotesize  \text{(by functoriality of $\tensor$ and $A \tensor \unit = A$)}}
\\
&=&
    \interp{o \; (d_{w_n} \; \ldots (d_{w_1} \; d_\epsilon) \ldots)}
    &
    {\footnotesize    \text{(by definition of $\interp{-}$)}}
\end{array}
\]
\end{proof}

\subsection{From first-order transductions to $\lampa$}
\label{subsec:fotrans2lambda}

Now we wish to prove the converse direction of Theorem~\ref{thm:main-planar-str}, that
is that every FO-transduction can be encoded in $\lampa$.
Much like in~\cite{aperiodic,planartrans}, we rely on the fact that affine $\lampa$-definable
string-to-string functions are closed under composition.
Using this and the seminal Krohn-Rhodes decomposition theorem, it
was already shown that affine $\lampa$-definable functions include all \emph{sequential
functions}~\cite[Theorem 5.4]{aperiodic}. We thus rely on the same strategy
that is used in~\cite{planartrans} to show that 2PRFTs compute all first-order transductions.

\begin{lemma}[{rephrasing of~\cite[Lemma~4.8]{ListFunctions}, see also~\cite[Lemma 4.3]{planartrans}}]
  \label{lem:fotransKR}
  Every first-order transduction can be decomposed as $f \circ \ttreverse \circ g \circ
  \ttreverse \circ h$ where $f$ is computed by a \emph{monotone register
  transducer} and the functions $g$ and $h$ are aperiodic sequential.
\end{lemma}

Example~\ref{ex:lam-rev} already shows that $\ttreverse$ is affine $\lampa$-definable.
Now it only remains to show that functions computed by monotone register transducers~\cite{SST}
are affine $\lampa$-definable.
Those machines go through their inputs in a single left-to-right pass,
storing infixes of their outputs in registers that they may update by
performing concatenations of previously stored values and constants.
\emph{Monotone} here corresponds to the further restrictions that those
machines have no control states, that the output corresponds to a single
register and that the register updates satisfy a monotonicity condition in
addition to being copyless.

First, let us define the notion of update those machines can use. For simplicity,
throughout the rest of this section we assume a fixed output alphabet $\Gamma$,
disjoint from the set of natural numbers, and a fixed input alphabet $\Sigma$.
\begin{definition}
The set of \emph{copyless monotone register update} from $n$ registers to $k$
registers, which we write $\CMRupdate{n}{k}$, is the subset
consisting of those $k$-uples $(w_0, \ldots, w_{k-1})$ of words over
$\Gamma \cup \{0, \ldots, n-1\}$ such that:
\begin{itemize}
  \item every index $i < n$ occurs at most once in the overall tuple \hfill {\footnotesize(copylessness/affineness)}
  \item if we have that $i \le j < n$ occurring in $w_{\ell_i}$ and $w_{\ell_j}$
    respectively, then we have either that $\ell_i < \ell_j$, or $\ell_i = \ell_j$
    and $i$ occurs before $j$ in $w_{\ell_i}$. \hfill {\footnotesize (monotonicity/planarity)}
\end{itemize}

Given $\sigma \in \CMRupdate{k}{\ell}$ and $\sigma' \in \CMRupdate{n}{k}$,
the composition $\sigma \circ \sigma' \in \CMRupdate{n}{\ell}$ is defined by
substituting each index $i < k$ in $\sigma$ by the $i$th component of $\sigma'$
(this preserves copylessness and monotonicity).
\end{definition}
At the intuitive level, an element of $\CMRupdate{n}{k}$ encodes a function
$(\Gamma^*)^n \to (\Gamma^*)^k$ that can operate by concatenating together the
components of its inputs, subject to restrictions that match affineness and planarity\footnote{This could have alternatively been defined as a free
affine strict monoidal category with a monoid object and generators for the letters
of $\Gamma$.}.
For the sequel, write
$\pi_\ell \in \CMRupdate{k}{1}$ for $\ell < k$ for the obvious projections, $\epsilon^k$ for the updates of $\CMRupdate{0}{k}$
that initialize every register with the empty word and
$\registerContent$ for the canonical
isomorphism $\registerContent : \CMRupdate{0}{1} \cong \Gamma^*$.
With this in hand, we 
give a working definition of monotone register transducers.

\begin{definition}
  A \emph{monotone register transducer}
  consists of the following:
  \begin{itemize}
  \item a number $n$ of registers
  \item for each input letter $a\in\Sigma$, a copyless monotone register update $\sigma_a : x^n \to x^n$.
  \end{itemize}
  It computes the function
$
  \begin{array}{ccc}
  \Sigma^* & \longto & \Gamma^* \\
  a_1 \ldots a_n &\longmapsto & \registerContent(\pi_0 \circ \sigma_{a_n} \circ \ldots \circ \sigma_{a_1} \circ \epsilon^k) \\
  \end{array}
$
\end{definition}

Now we will argue that for every monotone register transducer with $n$ registers, we can produce an equivalent $\lampa$-transducer
with some iteration type $\kappa_n \lin \basety$.
The intuition behind the definition of $\kappa_n$ is that a register holding a string that support concatenations
can be encoded using the type $\basety \lin \basety$ and composition.
As we need $n$ copies of those, we thus set
\[\kappa_n = \underbrace{(\basety \lin \basety) \lin \ldots \lin (\basety \lin \basety)}_{\text{$n$-fold}} \lin \basety\]
so that $\kappa_n \lin \basety$ is a sufficiently expressive
stand-in for the $n$-fold tensor of $\basety \lin \basety$.

\begin{lemma}
\label{lem:CMRupdate-lambda}
Every $\sigma \in \CMRupdate{k}{n}$ maps to a $\lampa$-term
$\underline{\Gamma}; \cdot \vdash \underline{\sigma} : \kappa_n \lin \kappa_k$
in a way that is compatible with composition, that is
$\underline{\sigma \circ \sigma'} =_{\beta\eta} \lambda z. \; \underline{\sigma'} \; (\underline{\sigma} \; z)$.
Finally, if $\sigma \in \CMRupdate{0}{1}$, we have
  $\underline{\registerContent(\sigma)} =_{\beta\eta} \underline{\sigma} \; (\lambda x. x)$.
\end{lemma}

\begin{proof*}{Proof idea}
For 
$\sigma = (w_1, \ldots, w_n) \in \CMRupdate{k}{n}$, define
$\underline{\sigma} : \kappa_n \to \kappa_k$ to be the term $\lambda F \; f_1 \ldots f_k. F \; t_1 \ldots t_n$
where $t_i$ is obtained by recursion over $w_i$, starting with the identity and postcomposing with
\begin{itemize}
  \item the appropriate  constant from $\underline{\Gamma}$ when we encounter a letter of $\Gamma$
  \item $f_k$ if we encounter the index $k$
\end{itemize}
This is typable in $\lampa$ specifically because the transitions are monotone and copyless.
Then it is relatively straightforward to check that we have the advertised equations.
\end{proof*}

Then the $\lampa$-terms corresponding to transitions will essentially precompose
the suitable terms $\underline{\sigma}$ defined in Lemma~\ref{lem:CMRupdate-lambda}.
This corresponds to applying the exponentiation operation, defined
by $t \lin \basety = \lambda X. \lambda z. \; X \; (t \; z)$.
This operation is compatible with composition, i.e. we have
$(t \lin \basety) \circ (u \lin \basety) =_{\beta\eta} (u \circ t) \lin \basety$
for arbitrary terms $t$ and $u$ which make those expressions typecheck.

\begin{lemma}
Every function definable by a monotone register transducer is $\lampa$-definable.
\end{lemma}
\begin{proof*}{Proof idea}
  Suppose we are given such a transducer with $n$ registers and transitions
  $(\sigma_a)_{a \in \Sigma}$ and let us build terms as per Definition~\ref{def:normalAffineStr2Str}.
  We take for iteration type $\kappa_n \lin \basety$, $d_a = \underline{\sigma_a} \lin \basety$,
  $d_\epsilon = \lambda Z. Z \; (\lambda x. x) \; \ldots \; (\lambda x. x)$
  and $o = \lambda K. (K \circ \underline{\pi_0}) \; (\lambda Z. Z \; (\lambda x. x))$.
  Then using Lemma~\ref{lem:CMRupdate-lambda}, we can check step-by-step we have the desired equations.
\end{proof*}

\section{Conclusion}
\label{sec:conc}

We have now proven that affine $\lampa$-definable string-to-string functions
correspond exactly to first-order transductions. One key aspect of the proof
was to use a semantic interpretation of purely affine $\lambda$-terms as planar diagrams to
compile $\lampa$-transducers to 2PRFTs. This result essentially closes the open questions
raised in~\cite{aperiodic} and provides an alternative, less syntactic, proof for the
soundness part of its main theorem.

We will now discuss further results that could be derived by adapting
the material we have developed in the previous section. We will then
list some questions that arise because of, or could be solved using,
the interpretation of terms as (planar) diagrams.

\subsection{Discussion on variations: dropping planarity, regular transductions \& tree languages}

A natural variation on $\catTransition_\Gamma$ is to drop the planarity
requirement on the morphisms so that wires may cross in the geometric realizations
of diagrams. If we do so, the tensor product becomes \emph{symmetric}, that is
we have a natural isomorphisms $\gamma_{A,B} : A \tensor B \to B \tensor A$
such that $\gamma_{A,B} = \gamma^{-1}_{B,A}$ and
$\gamma_{A, \unit} = \id_A$\footnote{This is of course assuming a strict monoidal product.},
while still keeping a poset-enriched autonomous structure.
This change makes the order of nodes in diagrams irrelevant, and objects
with the same number of $+$ and $-$ occurring isomorphic\footnote{Quotienting sensibly
yields a (poset-enriched) category isomorphic to the one computed by applying the
$\Int$ construction~\cite[\S4]{jsvtraced}
to the category whose objects are natural numbers regarded as finite sets and morphisms
from $n$ to $k$ are the subsets of $n \times \Gamma^* \times k$ that induce partial
injections from $n$ to $k$. The composition is then defined by $f \circ g = \{ (i, uv, j) \mid \exists \ell.~ (i,u,\ell) \in f \wedge (\ell, v, j) \in
g\}$ and then the traced monoidal structure is defined analogously to that of
the category of partial injections.}.
This allows to model the commutative variation of $\lampa$, which we call $\lama$,
where we include the exchange rule:
\[
\dfrac{\Gamma; \; \Delta, y : \tau_2, x : \tau_1, \Delta' \vdash t : \sigma}
{\Gamma; \; \Delta, x : \tau_1, y : \tau_2, \Delta' \vdash t : \sigma}\]
If we define what are (affine) $\lama$-definability and $\lama$-transducers in
a manner analogous to $\lampa$-definability and $\lampa$-transducers,
as well as the notion of (not necessarily planar)
\emph{two-way reversible finite transducers} (2RFTs, which match the notion in~\cite{reversibletransducers} and thus capture all \emph{regular}
transductions), we have the following.

\begin{theorem}
\label{thm:commutative}
Affine $\lama$-definable functions and regular transductions coincide:
\begin{itemize}
  \item $\lama$-transducers can be translated into equivalent 2RFTs in polynomial time
  \item regular transductions are $\lama$-definable
\end{itemize}
\end{theorem}
\begin{proof*}{Proof idea}
  The first point is obtained by an easy adaptation the arguments of Subsections~\ref{subsec:interplambda}
  (where we add the interpretation of the exchange rule using the symmetry $\gamma$)
  and~\ref{subsec:lambda22PRFT}. The second point is also obtained by an
  argument similar to the one in Subsection~\ref{subsec:fotrans2lambda}: Lemma~\ref{lem:fotransKR}
  holds if we replace ``aperiodic sequential'' by ``sequential'' and ``first-order transduction''
  by ``regular transduction''. We then need to know that all sequential functions are $\lama$-definable, which is true by~\cite[Theorem 5.4]{aperiodic}.
\end{proof*}

This statement should be contrasted with~\cite[Theorem 1.1]{freeadditives}
which states that regular string-to-string transductions coincide with functions
definable in a variant of $\lama$ which is augmented with additives\footnote{And also linear instead of affine, however in the presence of additives, this
distinction is not very important (see~\cite[\S1.2.1]{LLSS} for a discussion).}.
There, $\lambda$-terms defining string-to-string functions are compiled into
streaming string transducer (SSTs). But this translation can yield a
machine that has a state-space whose size is non-elementary in terms of the size
of an input $\lama$-transduction free of additives connectives.
Since the translations between 2RFTs and SSTs is $\ELEMENTARY$~\cite{reversibletransducers},
the translation we offer here is more efficient.
On the other hand, the second point improves on~\cite{freeadditives} by
compiling first-order transductions in a smaller $\lambda$-calculus at the
cost of employing Lemma~\ref{lem:fotransKR} that relies on the powerful
and relatively complex technique of Krohn-Rhodes decomposition instead of a
direct polynomial-time compilation of SSTs.

While we have only investigated functions that take strings as inputs in this
paper, the tools we have introduced can be used to study functions
that take ranked trees as input (and still output strings). Indeed, ranked
trees, that are parameterized by finite ranked alphabet, can be represented by
Church encodings and given
precise affine typing (c.f.~\cite[\S2.3]{freeadditives}).
In that case, $\lama$-terms get compiled to what amounts to
\emph{reversible tree-walking transducers with string output}
(or simply reversible tree-walking automata if we take the output alphabet to be empty)
as defined by restricting Definitions 3.5 and 3.8 of~\cite{nguyenvanoni} to
string outputs. As a result, we can give a new proof of the following theorem,
which is also a consequence of~\cite[Theorem~1.4]{nguyenvanoni}\footnote{Both arguments
essentially appeal to Girard's geometry of interaction,
but theirs is based on compiling executions of an abstract machine evaluating
$\lambda$-terms while we focus on a semantic interpretation of linear logic.}.

\begin{theorem}
  \label{thm:twt}
Every $\lama$ tree-to-string transducer can be turned into an equivalent reversible tree-walking transducer.
\end{theorem}

This result means the affine $\lambda$-calculus without
additives cannot recognize all regular tree languages~\cite{bojanczyk2008tree},
whereas allowing additives 
captures all regular tree transductions~\cite[Theorem 1.2]{freeadditives}.

\subsection{Perspectives}

A natural question is whether Theorem~\ref{thm:twt} admits a converse: is every
language recognized by a reversible tree-walking automaton also recognized by some $\lama$-term?
Another natural question is ``what are the tree languages recognized by $\lampa$-terms?''.
Clearly, they should be recognized by tree-walking automata that are not only
\emph{reversible}, but also \emph{planar} in the obvious sense. This is an
actual restriction, as a non-planar tree-walking automaton could
count the number of leaves modulo 2, which a planar device could not. So we
can also ask the question: is every language recognized by a planar reversible
tree-walking automaton also recognized by some $\lampa$-term?
These questions might be challenging since we are currently not aware of a
convenient tool similar to the Krohn-Rhodes theorem or~\cite[Theorem 3.4]{fotree}
that would allow to decompose tree-walking transducers.
A first step might be to check that those transducers, as well as their planar
variant, are closed under composition. This would require considering tree-to-tree
transductions as discussed in~\cite{nguyenvanoni}, which would naturally
lead to extending our diagrammatic constructions so that they may
depend on a ranked alphabet, much like the categories of register updates
considered in~\cite{freeadditives}. A variant of the operad of spliced arrows
specified in~\cite[Definition 1.1]{MelliesZeilbergerContours} could be of use.

We have treated only \emph{affine} $\lampa$-definable functions in this paper.
The next question is whether we can also get a characterization of $\lampa$-definable
functions implemented by terms of type $\Str_\Sigma[\kappa] \to \Str_\Gamma$.
It is plausible they correspond to \emph{first-order polyblind} functions alluded
to in~\cite{comparisonfree}\footnote{In which they were called first-order comparison-free. We follow the terminological
change introduced in~\cite{Doueneau-Tabot22}.}, which are obtained by
closing first-order transductions under 
\emph{compositions by substitution}~\cite[Definition 4.1]{comparisonfree}.
Our hope is that this correspondence can be established using a similar strategy
as~\cite[\S5.3]{titophd}.

\bibliographystyle{./entics}
\bibliography{bi}
\newpage
\appendix
\section{$(\catTransition_\Gamma, +, +)$-automata vs 2PRFTs}

Let us detail here why
$(\catTransition_\Gamma, +, +)$-automata 
and 2PRFTs with output alphabet $\Gamma$ define the same string-to-string functions.

First, let us note that both options define an output in $\Gamma_\bot^*$ by
examining, when it exists, the label of the single possible edge of $\Hom{\catTransition_\Gamma}{+}{+}$
and $\Hom{\catTransition_\Gamma}{\unit}{+-}$. This operation
commutes with the isomorphism
\[\begin{array}{ccc}
\Hom{\catTransition_\Gamma}{+}{+} &\cong& \Hom{\catTransition_\Gamma}{\unit}{+-} \\
  f &\longmapsto& \varepsilon_+ ; (f \tensor \id_-) \\
  (g \tensor \id_+); (\id_+ \tensor \eta_+) &\longmapsfrom& g
\end{array}
\]

For the purpose of this discussion, fix an input alphabet $\Sigma$.
First assume we are given a $(\catTransition_\Gamma, +, +)$-automaton $\cA$.
We define an equivalent 2PRT (i.e. a $(\catTransition_\Gamma, \emptyString, +-)$-automaton)
$\cA'$ by setting
\[\cA'(\sState) = \cA(\sState) \tensor - \qquad
\cA'(\triangleright) = \varepsilon_+; (\cA(\triangleright) \tensor \id_-)
\qquad \text{and otherwise} \qquad
\cA'(f) = \cA(f) \tensor \id_-\]
which is easily checked to be functorial and is equivalent to $\cA$ as we have
{\small
\[
\begin{array}{lcll}
(\cA'(\triangleright w \triangleleft) \tensor \id_+); (\id_+ \tensor \eta_+)
  &=&
((\cA'(\triangleright);\cA'(w\triangleleft)) \tensor \id_+); (\id_+ \tensor \eta_+)
  & \text{by functoriality of $\cA'$}\\
  &=&
((\varepsilon_+; (\cA(\triangleright) \tensor \id_-);(\cA(w\triangleleft) \tensor \id_-)) \tensor \id_+); (\id_+ \tensor \eta_+)
  & \text{by definition of $\cA'$}\\
  &=&
  (\varepsilon_+ \tensor \id_+); ((\cA(\triangleright);\cA(w\triangleleft)) \tensor \id_{-+}); (\id_+ \tensor \eta_+)
  & \text{by functoriality of $\tensor$}
  \\
  &=&
  (\varepsilon_+ \tensor \id_+);(\id_+ \tensor \eta_+);
  (\cA(\triangleright);\cA(w\triangleleft))
  & \text{by naturality of $\eta$}\\
  &=&
  \cA(\triangleright);\cA(w\triangleleft)
& \text{by a zigzag equation}\\
  &=&
  \cA(\triangleright w\triangleleft)
& \text{by functoriality}\\
\end{array}
\]
}
Conversely, if we have a 2PRFT $\cT$, we can turn it into an equivalent
$(\catTransition_\Gamma, +, +)$-automaton $\cT'$ by setting
\[\cT'(\sState) = \cT(\sState) \tensor + \qquad
\cT'(\triangleleft) = \cT(\triangleleft) \tensor \eta_+
\qquad \text{and otherwise} \qquad
\cT'(f) = \cA(f) \tensor \id_+\]
The proof that it is equivalent to $\cT$ is similar to the one above,
exploiting the naturality of $\varepsilon$ and the other zigzag equation.
\end{document}